\documentclass[10pt,a4paper,draft]{article} 

\usepackage{lmodern}
\usepackage{lmodern}
\usepackage[T1]{fontenc}
\usepackage{hyperref}
\usepackage{units}
\usepackage{amsmath}
\usepackage{amsthm}
\usepackage{amssymb}
\usepackage{graphicx}
\usepackage{esint}
\usepackage{natbib}

\usepackage{epsfig}
\usepackage{pgf}
\usepackage{tikz}\usepackage{mathrsfs}
\usetikzlibrary{arrows}

\newtheorem{Prop}{Proposition}
\newtheorem{lem}{Lemma}
\newtheorem{cor}{Corollary}
\newtheorem{thm}{Theorem}
\newtheorem{Definition}{Definition}
\newtheorem{Example}{Example}

\begin{document}

\title{Divergence and Sufficiency for Convex Optimization}
\author{Peter Harremo{\" e}s}
\maketitle

\abstract{Logarithmic score and information divergence appear in
information theory, statistics, statistical mechanics, and
portfolio theory. We
demonstrate that all these topics involve some kind of
optimization that leads directly to regret functions and such regret functions are often given by a Bregman
divergence. If the regret function
also fulfills a sufficiency condition it must be proportional to
information divergence. We will demonstrate that sufficiency is
equivalent to the
apparently weaker notion of locality and it is also equivalent
to the apparently stronger notion of monotonicity. These
sufficiency conditions have
quite different relevance in the different areas of application,
and often they are not fulfilled. Therefore sufficiency
conditions can be used to
explain when results from one area can be transferred directly
to another and when one will experience differences.}

\section{Introduction}

One of the main purposes of information theory is to compress
data so that data can be recovered exactly or approximately. One of
the most important quantities was called entropy because it is
calculated according to a formula that mimics the calculation of entropy in
statistical mechanics. Another key concept in information theory is
information divergence (KL-divergence) that is defined for probability vectors $P$ and $Q$ as
\[
D\left(P\Vert Q\right)=\sum_x P(x)\ln\frac{P(x)}{Q(x)}.
\]
It was introduced by Kullback and
Leibler in 1951 in a paper entitled information and sufficiency
\cite{Kullback1951}. The link
from information theory back to statistical physics was
developed by E.T. Jaynes via the maximum entropy principle \cite{Jaynes1957,Jaynes1989}. The link back
to statistics is now well established
\cite{Liese1987,Barron1998,Csiszar2004,Grunwald2004a,Grunwald2007}.

Related quantities appear in information theory, statistics,
statistical mechanics, and finance, and we are interested in a
theory that describes when these relations are exact and when they just work by
analogy. First we introduce some general results about optimization on
state spaces of finite dimensional C*-algebras. This part applies exactly to all the topics under
consideration and lead to Bregman divergences. Secondly, we introduce several notions of
sufficiency and show that this leads to information divergence.
This second step is not always applicable which explains when the
different topics are really different.

\section{Structure of the state space}

Our knowledge about a system will be represented by a state
space. I many cases the state space is given by a set of probability
distributions on a sample space. In such cases the state space is a simplex,
but it is well-known that the state space is not a simplex in
quantum physics. For applications in quantum physics the state space is
often represented by a set of density matrices, i.e. positive
semidefinite complex matrices with trace 1. In some cases the states are represented
as elements of a finite dimensional $C^{*}$-algebra, which is a direct sum of matrix
algebras. A finite dimensional $C^{*}$-algebra that is a sum of $1\times1$ matrices has a state space that is a simplex, so the state spaces of finite dimensional $C^{*}$-algebras contain the classical probability distributions as special
cases.

The extreme points in the set of states are the pure states. The
pure states of a $C^{*}$-algebra can be identified with projections
of rank 1. Two density matrices $s_{1}$ and $s_{2}$ are said to be
orthogonal if $s_{1}s_{2}=s_{2}s_{1}=0.$ Any state $s$ has a
decomposition
\[
s=\sum\lambda_{i}s_{i}
\]
where $s_{i}$ are orthogonal pure states. Such a decomposition
is not unique, but for a finite dimensional $C^{*}$-algebra the coefficients
$\lambda_{1},\lambda_{2},\dots,\lambda_{n}$
are unique and are called the spectrum of the state. 

Sometimes more general state spaces 
are of interest. In generalized probabilistic theories a state space is a convex set where
mixtures are defined by randomly choosing certain states with
certain probabilities \cite{Holevo1982,Krumm2016}. A convex set where all orthogonal
decompositions of a state have the same spectrum is called a spectral state
space. Much of the theory in this paper can be generalized to spectral sets.
The most important spectral sets are sets of positive elements with  trace 1
in Jordan algebras. For questions related to the foundation of
quantum theory the Jordan algebras and other spectral sets give new
insight \cite{Barnum2014,Harremoes2016d,Harremoes2017a},
but in this paper we will restrict our attention to states on
finite dimensional $C^{*}$-algebras. Nevertheless some of the
theorems and proofs are
stated in such a way that they hold for more general state
spaces.

\section{Optimization}

Let $\mathcal{S}$ denotes a state space of a finite dimensional $C^*$-algebra and let $\mathcal{A}$ denote a set of self-adjoint operators. Each $a\in\mathcal{A}$ is identified with a real valued measurement. The elements of $\mathcal{A}$ may represent feasible \emph{actions} 
(decisions) that lead to a payoff like the score of a
statistical decision, the energy extracted by a certain interaction with the
system, (minus) the length of a codeword of the next encoded input
letter using a specific code book, or the revenue of using a certain
portfolio. For each $s\in\mathcal{S}$ the mean value of the measurement $a\in\mathcal{A}$ is given by 
\[
\left\langle a,s\right\rangle =\rm{tr}(as).
\]
In this way the set of actions may be identified with a subset of the dual space of $\mathcal{S}$. Next we define 
\[
F\left(s\right)=\sup_{a\in\mathcal{A}}\left\langle
a,s\right\rangle .
\]
 We note that $F$ is convex, but $F$ need not be strictly convex. In principle $F(s)$ may be infinite but we will assume that $F(s)<\infty$ for all states $s$. We also note that $F$ is lower semi-continuous. In this paper we will assume that the function $F$ is continuous. The assumption that $F$ real valued continuous function is fulfilled for all the applications we consider. 

If $s$ is a state and $a\in\mathcal{A}$ is an action then we say that $a$ is {\em optimal} for $s$ if $\left\langle
a,s\right\rangle = F\left(s\right)$. A sequence of actions $a_n\in\mathcal{A}$  is said to be {\em asymptotically optimal} for the state $s$ if  $\left\langle a,s\right\rangle \to F\left(s\right)$ for $n\to \infty.$

If $a_{i}$ are actions and $\left(t_{i}\right)$ is a probability
vector then we we may define the mixed action $\sum t_{i}\cdot a_{i}$
as the action where we do the action $a_{i}$ with probability $t_{i}.$ We note
that $\left\langle \sum t_{i}\cdot a_{i},s\right\rangle =\sum t_{i}\cdot\left\langle a_{i},s\right\rangle .$
We will assume that all such mixtures of feasible actions are also
feasible. If $a_{1}\left(s\right)\geq a_{2}\left(s\right)$ almost
surely for all states we say that $a_{1}$ dominates $a_{2}$ and
if $a_{1}\left(s\right)>a_{2}\left(s\right)$ almost surely for all
states $s$ we say that $a_{1}$ strictly dominates $a_{2}.$ All
actions that are dominated may be removed from $\mathcal{A}$ without
changing the function $F.$  Let $\mathcal{A}_{F}$ denote the set of self-adjoint operators (observables) such that $\left\langle m,s\right\rangle \leq F\left(s\right).$
Then $F\left(s\right)=\sup_{a\in\mathcal{A}_{F}}\left\langle a,s\right\rangle .$
Therefore we may replace $\mathcal{A}$ by $\mathcal{A}_{F}$ without
changing the optimization problem.

In the definition of regret we follow Servage \cite{Servage1951}
but with different notation.
\begin{Definition}
Let $F$ denote a convex function on the state space
$\mathcal{S}$. If $F\left(s\right)$ is finite \emph{the regret}
of the action $a$ is defined by 
\begin{equation}
D_{F}\left(s,a\right)=F\left(s\right)-\left\langle
a,s\right\rangle .
\end{equation}
\end{Definition}
\begin{Prop}\label{prop:actionregret}
The regret $D_{F}$ of actions has the following properties:
\begin{itemize}
\item $D_{F}\left(s,a\right)\geq0$ with equality if $a$ is
optimal for $s$.
\item $s\to D_{F}\left(s,a\right)$ is a convex function.
\item If $\bar{a}$ is optimal for the state $\bar{s}=\sum
t_{i}\cdot s_{i}$
where $\left(t_{1},t_{2},\dots,t_{\ell}\right)$ is a probability vector then 
\[
\sum t_{i}\cdot D_{F}\left(s_{i},a\right)=\sum t_{i}\cdot
D_{F}\left(s_{i},\bar{a}\right)+D_{F}\left(\bar{s},a\right).
\]
 
\item $\sum t_{i}\cdot D_{F}\left(s_{i},a\right)$ is minimal if
$a$ is optimal for $\bar{s}=\sum t_{i}\cdot s_{i}$.
\end{itemize}
\end{Prop}

\begin{figure}
\begin{centering}

\begin{tikzpicture}[line width=1.2pt, line cap=round,line join=round,>=triangle 45,x=10.0cm,y=9.999999999999998cm, scale=0.8]
\draw[->,color=black] (0.,0.) -- (1.15,0.);
\foreach \x in {0,1}
\draw[shift={(\x,0)},color=black] (0pt,2pt) -- (0pt,-2pt) node[below] {$\x$};
\draw[->,color=black] (0.,0.) -- (0.,0.7);

\draw [color=green] (0.,0.5)-- (1.,0.375);

\draw [dash pattern=on 4pt off 4pt] (1.,0.)-- (1.,0.375);
\draw  (1.,0.375)-- (1.,0.625);

\draw[color=blue,smooth,samples=100,domain=0:1] plot(\x,{(\x)^2/4 - (\x)/8 + 1/2});

\draw (1.02,0.52) node[anchor=north west] {$D_F(s_1,s_0)$};
\draw  [color=blue] (0.04,0.7) node[anchor=north west] {$f(t)$};
\draw (1.0807407407407414,0.1) node[anchor=north west] {$t$};

\draw [line width=0.3pt, fill=blue] (0.,0.5) circle (2pt);
\draw [line width=0.3pt, fill=blue] (1.,0.625) circle (2pt);
\draw [fill=black] (1.,0.) circle (2pt);
\draw [line width=0.3pt, fill=green] (1.,0.375) circle (2pt);
\end{tikzpicture}
\par\end{centering}
\caption{The regret equals the vertical distance between curve
and tangent.\label{Fig:breg}}
\end{figure}
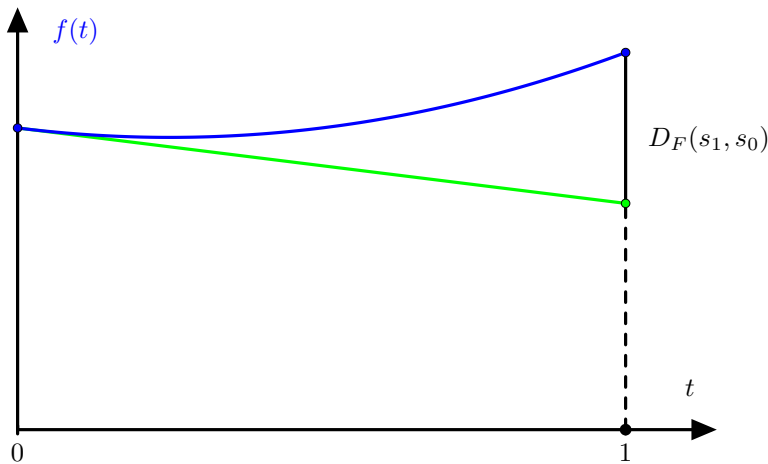

If the state is $s_{1}$ but one acts as if the state were
$s_{0}$ one may compare what one
achieves and what could have been achieved. If the state $s_0$ has a unique optimal action $a$ we may simply define the regret of $s_0$ by
\[
D_{F}\left(s_{1},s_{0}\right)=D_{F}\left(s_{1},a\right)
\]
The following definition leads to a regret function that is essentially equivalent to the so-called {\em generalized Bregman divergences} defined by Kiwiel \cite{Kiwiel1997,Kiwiel1997a}. 

\begin{Definition}\label{def:regret} 
Let $F$ denote a convex function on the
state space $\mathcal{S}$. If $F\left(s_{1}\right)$ is
finite then we define \emph{the regret of the state} $s_{0}$ as 
\[
D_{F}\left(s_{1},s_{0}\right)=\inf_{\left(a_n\right)}\lim_{n\to \infty}D_{F}\left(s_{1},a\right)
\]
where the infimum is taken over all sequences of actions $\left(a_n\right)$ that are asymptotically optimal for $s_{0}.$ 
\end{Definition}

With this definition the regret is always defined with values in $[0,\infty$. We note that with this definition the value of the regret $D_{F}\left(s_{1},s_{0}\right)$ only depends on the restriction of the function $F$ to the line segment from $s_0$ to $s_1$. Let $f$ denote the function $f(t)=F\left((1-t)s_0 +t s_1 \right)$ where $t\in[0,1]$. As illustrated in Figure \ref{Fig:breg} we have
\begin{equation}\label{eq:rightderiv}
D_F\left(s_1 ,s_0 \right)=f\left(1 \right) - \left( f\left(0\right)+ f'_+ (0) \right)
\end{equation}
where $f'_+ (0)$ denotes the right derivative of $f$ at $t=0$. Equation (\ref{eq:rightderiv}) is even valid when the regret is infinite if we allow the right derivative to take the value $\infty$.

If the state $s_{0}$ has the unique optimal action $a\in\mathcal{A}$ then
\begin{equation}\label{eq:rekonst}
F\left(s_{1}\right)=D_{F}\left(s_{1},s_{0}\right)+\left\langle
a,s_{1}\right\rangle
\end{equation}
so the function $F$ can be reconstructed from $D_{F}$ except for an affine function of $s_{1}.$ The closure of the convex hull of the set of functions $s\to\left\langle a,s\right\rangle $ is
uniquely determined by the convex function $F.$ The following proposition follows from Alexandrov's theorem. See
\cite[Theorem 25.5]{Rockafeller1970} for details.
\begin{Prop}
A convex function on a finite dimensional convex set is
differentiable almost everywhere with respect to the Lebesgue measure. 
\end{Prop}
A state $s_0$ where $F$ is differentiable has a unique optimal action. Therefore Equation (\ref{eq:rekonst}) holds for almost any state $s_0$. In particular the function $F$ can be reconstructed from $D_F$ except for an affine function.

\begin{Prop}\label{prop:regretproperties}
The regret $D_{F}$ of states has the following properties:
\begin{itemize}
\item $D_{F}\left(s_{1},s_{0}\right)\geq0$ with equality if
there exists an action $a$ that is optimal for both $s_1$ and
$s_0$.
\item $s_1 \to D_{F}\left(s_1 ,s_0 \right)$ is a convex
function.
\end{itemize}
Further the following two conditions are equivalent.
\begin{itemize}
\item $D_{F}\left(s_{1},s_{0}\right)= 0$ implies $s_1 =s_0$ .
\item The function $F$ is strictly convex.
\end{itemize}
\end{Prop}
We say that a regret function $D_F$ is {\em strict} if $F$ is strictly convex. The two last properties Proposition \ref{prop:actionregret} do not carry over to regret for states except if the regret is a {\em Bregman
divergence} as defined below.
The regret is called a \emph{Bregman divergence} if it can be written in the following form
\begin{align}
D_{F}\left(s_{1},s_{0}\right) &
=F\left(s_{1}\right)-\left(F\left(s_{0}\right)+\left\langle
s_{1}-s_{0},\nabla F\left(s_{0}\right)\right\rangle \right)
\end{align}
where $\left\langle \cdot,\cdot\right\rangle $ denotes the (Hilbert-Smidt) 
inner product. In the context of forecasting and
statistical scoring rules the use of Bregman divergences dates back to
\cite{Hendrickson1971}. A similar but less general definition of regret was given by Rao
and Nayak \cite{Rao1985} where the name \emph{cross entropy} was
proposed. Although Bregman divergences have been known for many years they did not gain popularity before the paper \cite{Banerjee2005} where a systematic study of Bregman divergences was presented.

We note that if $D_{F}$ is a Bregman divergence and $s_{0}$
minimizes $F$ then $\nabla F\left(s_{0}\right)=0$ so that the formula for
the Bregman divergence reduces to 
\[
D_{F}\left(s_{1},s_{0}\right)=F\left(s_{1}\right)-F\left(s_{0}\right).
\]
Bregman divergences satisfy the \emph{Bregman identity} 
\begin{equation}\label{eq:Bregmanid}
\sum t_{i}\cdot D_{F}\left(s_{i},s\right)=\sum t_{i}\cdot
D_{F}\left(s_{i},\bar{s}\right)+D_{F}\left(\bar{s},s\right),
\end{equation}
but if $F$ is not differentiable this identity can be violated.
\begin{Example}
Let the state space be the interval $\left[0,1\right]$ with two
actions $\left<a_{0},s\right> = 1-2s$ and $\left<a_{1},s\right> =2s-1.$ Let
$s_{0}=0$ and $s_{1}=1.$ Let further $t_{0}=\nicefrac{1}{3}$ and
$t_{1}=\nicefrac{2}{3}.$ Then $\bar{s}=\nicefrac{2}{3}.$ If $s=\nicefrac{1}{2}$ then 
\[
\sum t_{i}\cdot D_{F}\left(s_{i},s\right)=0,
\]
but 
\begin{eqnarray*}
\sum t_{i}\cdot D_{F}\left(s_{i},\bar{s}\right) & = &
\frac{1}{3}\cdot\left(\left<a_{0},0\right> -\left<a_{1},0\right>\right)+\frac{2}{3}\cdot\left(\left<a_{1},1\right> -\left<a_{1},1\right>\right)\\
 & = & \frac{1}{3}\cdot\left(1-\left(-1\right)\right)\\
 & = & \frac{2}{3}.
\end{eqnarray*}
Clearly the
Bregman identity (\ref{eq:Bregmanid}) is violated and $\sum t_{i}\cdot D_{F}\left(s_{i},s\right)$ will increase if $s$ is replaced by $\bar s $.
\end{Example}
The following proposition is easily proved.
\begin{Prop}
For a convex and continuous function $F$ the following conditions are equivalent.
\begin{itemize}
\item The function $F$ is differentiable.
\item The regret $D_{F}$ is a Bregman divergence.
\item The Bregman identity is always satisfied.
\item For any probability vectors $\left(t_1 , t_2 ,\dots , t_n\right)$ the sum $\sum t_{i}\cdot D_{F}\left(s_{i},s\right)$ is always minimal when $s=\sum t_{i}\cdot s_i$ .
\end{itemize}
\end{Prop}

\section{Examples}

In this section we shall see how regret functions are defined in some applications.

\subsection{Information theory}

We recall that a code is uniquely decodable if any finite
sequence of input symbols give a unique sequence of output symbols. It is
well-known that a uniquely decodable code satisfies Kraft's inequality 
\begin{equation}
\sum_{a\in\mathbb{A}}\beta^{\textrm{-}\ell\left(a\right)}\leq1\label{eq:Kraft}
\end{equation}
where $\ell\left(a\right)$ denotes the length of the codeword
corresponding to the input symbol $a\in\mathbb{A}$ and $\beta$ denotes the
size of the output alphabet $\mathbb{B}$. Here the length of a codeword is an integer. If $P=\left(p_{a}\right)_{a\in\mathbb{A}}$ is a probability vector over the input alphabet, then the mean
code-length is 
\[
\sum_{a\in\mathbb{A}}\ell\left(a\right)\cdot p_{a}.
\]
Our goal is to minimize the expected code-length. Here the state
space consist of probability distributions over the input
alphabet and the actions are code-length functions.

Shannon established the inequality 
\[
-\sum_{a\in\mathbb{A}}\log_{b}\left(p_{a}\right)\cdot
p_{a}\leq\min\sum_{a\in\mathbb{A}}\ell\left(a\right)\cdot
p_{a}\leq-\sum_{a\in\mathbb{A}}\log_{b}\left(p_{a}\right)\cdot
p_{a}+1.
\]
It is a combinatoric problem to find the optimal code length
function. In the simplest case with a binary output alphabet the optimal
code-length function is determined by the Huffmann algorithm.

A code-length function dominates another code-length function if all letters have  it has shorter code-length. If a code-length function is not dominated by another
code-length function then for all $a\in\mathbb{A}$ the length is bounded by
$\ell\left(a\right)\leq\left|\mathbb{A}\right|-1.$
 For fixed alphabets $\mathbb{A}$ and $\mathbb{B}$
there exists only a finite number of code-length functions $\ell$ that satisfy Kraft's inequality and are not dominated by other
code-length functions that satisfying Kraft's inequality.

\subsection{Scoring rules}

The use of scoring rules has a long history in statistics. An
early contribution was the idea of minimizing the sum of square
deviations that dates back to Gauss and works perfectly for Gaussian
distributions. In the 1920s Ramsay and de Finetti proved versions of the Dutch
book theorem where determination of probability distributions were
considered as dual problems of maximizing a payoff function. Later it was
proved that any consistent inference procedure corresponds to
optimizing with respect to some payoff function. A more systematic study of
scoring rules was given by McCarthy \cite{McCarthy1956}.

Consider an experiment with $\mathcal{X}=\left\{
1,2,\dots,\ell\right\} $
as sample space. A \emph{scoring rule} $f$ is defined as a
function $\mathcal{X}\times
M_{1}^{+}\left(\mathcal{X}\right)\to\mathbb{R}$
such that the score is $f\left(x,Q\right)$ when a prediction has been given in terms of a probability distribution $Q$ and $x\in\mathcal{X}$ has been observed. A
scoring rule is \emph{proper} if for any probability measure $P\in
M_{1}^{+}\left(\mathcal{X}\right)$
the score $\sum_{x\in\mathcal{X}}P\left(x\right)\cdot
f\left(x,Q\right)$ is minimal when $Q=P.$ Here the state space consist of
probability distributions over $\mathcal{X}$ and the actions are
predictions over $\mathcal{X}$, which are also probability distributions over
$\mathcal{X}$.

There is a correspondence between proper scoring rules and
Bregman divergences as explained in \cite{Gneiting2007, Ovcharov2015}. 
If $D_F$ is a Bregman divergence and $g$ is a function with domain ${\mathcal X}$ then $f$ given by
$f\left(x,Q\right)=g\left(x\right) - D_{F}\left(\delta_{x},Q\right)$ defines a scoring rule. 

Assume that $f$ is a proper scoring function. Then a function $F$ can be defined as
\[
F(P)=\sum_{x\in\mathcal{X}}P\left(x\right)\cdot f\left(x,P\right)
\]
This lead to the regret function
\begin{equation}
D_F\left(P,Q\right)=F(P)-\sum_{x\in\mathcal{X}}P\left(x\right)\cdot
f\left(x,Q\right).
\end{equation}
Since $f$ is assumed to be proper $D_F \left(P,Q\right)\geq0$. The Bregman identity (\ref{eq:Bregmanid})
follows by straight forward calculations. With these two results
we see that the regret function $D_F$ is a Bregman divergence and that 
\begin{alignat}{1}
D_{F}\left(\delta_{y},Q\right) &
=\sum_{x\in\mathcal{X}}\delta_{y}\left(x\right)\cdot
f\left(x,\delta_{y}\right)\nonumber - \sum_{x\in\mathcal{X}}\delta_{y}\left(x\right)\cdot
f\left(x,Q\right) \\
 & =f\left(y,\delta_{y}\right) - f\left(y,Q\right).
\end{alignat}
Hence a proper scoring rule $f$ has the form 
$f\left(x,Q\right)= g(x) - D_{F}\left(\delta_{x},Q\right)$ where $g(x)=f\left(x,\delta_{x}\right)$. A {\em  strictly proper scoring rule} can be defined as a proper scoring rule where the corresponding Bregman divergence is strict.

\begin{Example}
The Brier score is given by 
\[
f\left(x,Q\right)=\frac{1}{n}\left(\sum_{y\in\mathcal{X}}\left(Q\left(y\right)-\delta_{x}\left(y\right)\right)^{2}\right).
\]
The Brier score is generated by the strictly convex function
$F\left(P\right)=\frac{1}{n}\sum_{x\in\mathcal{X}}P\left(x\right)^{2}$
\end{Example}

\subsection{Statistical mechanics}

Thermodynamics is the study of concepts like heat, temperature
and energy. A major objective is to extract as much energy from a
system as possible. The idea in statistical mechanics is to view the
macroscopic behavior of a thermodynamic system as a statistical consequence
of the interaction between a lot of microscopic components where
the interacting between the components are governed by very simple
laws. Here the central limit theorem and large deviation theory play a
major role. One of the main achievements is the formula for entropy as
a logarithm of a probability.

Here we shall restrict the discussion to the most simple kind of thermodynamic
system from which we want to extract energy. We may
think of a system of non-interacting spin particles in a magnetic
field. For such a system the Hamiltonian is given by 
\[
\hat{H}\left(\sigma\right)=-\mu\sum h_{j}\sigma_{j}
\]
where $\sigma$ is the spin configuration, $\mu$ is the magnetic
moment, $h_{j}$ is the strength of an external magnetic field, and
$\sigma_{j}=\pm1$
is the spin of the the $j$'th particle. If the system is in
thermodynamic equilibrium the configuration probability is
\[
P_{\beta}\left(\sigma\right)=\frac{\exp\left(-\beta
\hat{H}\left(\sigma\right)\right)}{Z_{\beta}}
\]
where $Z\left(\beta\right)$ is the partition function
\[
Z\left(\beta\right)=\sum_{\sigma}\exp\left(-\beta
\hat{H}\left(\sigma\right)\right).
\]
Here $\beta$ is the inverse temperature $\left(kT\right)^{-1}$
of the spin system and
$k=1.381\cdot10^{-23}\nicefrac{\unit{J}}{\unit{K}}$ is
Boltzmann's constant. 

The mean energy is given by 
\[
\sum_{\sigma}P_{\beta}\left(\sigma\right)\hat{H}\left(\sigma\right)
\]
which will be identified with the internal energy $U$ defined
in thermodynamics. The Shannon entropy can be calculated
as 
\begin{align*}
-\sum_{\sigma}P_{\beta}\left(\sigma\right)\ln
P_{\beta}\left(\sigma\right) &
=-\sum_{\sigma}P_{\beta}\left(\sigma\right)\ln\frac{\exp\left(-\beta
\hat{H}\left(\sigma\right)\right)}{Z_{\beta}}\\
& =-\sum_{\sigma}P_{\beta}\left(\sigma\right)\left(-\beta
\hat{H}\left(\sigma\right)-\ln Z\left(\beta\right)\right)\\
&
=\beta\cdot U +\ln
Z\left(\beta\right).
\end{align*}
The Shannon entropy times $k$  will be identified with the thermodynamic entropy $S$.

The amount of energy that can be extracted from the system if a
heat bath is available, is called the \emph{exergy}
\cite{Gundersen2011a}. We assume
that the heat bath has temperature $T_0$ and the internal energy
and entropy of the system are $U_0$ and $S_0$ if the system has
been brought in
equilibrium with the heat bath. The exergy can be calculated by
\begin{align*}
Ex & =U-U_{0}-T_{0}\left(S-S_{0}\right)\\
& =U-U_{0}-kT_{0}\left(\beta\cdot U+\ln
Z\left(\beta\right)-\beta_{0}U_{0}-\ln
Z\left(\beta_{0}\right)\right)\\
& =kT_{0}\left(\left(\beta_{0}-\beta\right)\cdot
U+\ln\frac{Z\left(\beta_{0}\right)}{Z\left(\beta\right)}\right).
\end{align*}
The information divergence between the actual state and the
corresponding state that is in equilibrium with the environment is 
\begin{align*}
D\left(\left.P_{\beta}\right\Vert P_{\beta_{0}}\right) &
=\sum_{\sigma}P_{\beta}\left(\sigma\right)\ln\frac{P_{\beta}\left(\sigma\right)}{P_{\beta_{0}}\left(\sigma\right)}\\
&
=\sum_{\sigma}P_{\beta}\left(\sigma\right)\ln\frac{\frac{\exp\left(-\beta
\hat{H}\left(\sigma\right)\right)}{Z\left(\beta\right)}}{\frac{\exp\left(-\beta_{0}\hat{H}\left(\sigma\right)\right)}{Z\left(\beta_{0}\right)}}\\
& =\sum_{\sigma}P_{\beta}\left(\sigma\right)\left(-\beta
\hat{H}\left(\sigma\right)+\beta_{0}\hat{H}\left(\sigma\right)+\ln\frac{Z\left(\beta_{0}\right)}{Z\left(\beta\right)}\right)\\
&
=\left(\beta_{0}-\beta\right)\cdot\sum_{\sigma}P_{\beta}\left(\sigma\right)\hat{H}\left(\sigma\right)+\ln\frac{Z\left(\beta_{0}\right)}{Z\left(\beta\right)}\\
& =\left(\beta_{0}-\beta\right)\cdot
U+\ln\frac{Z\left(\beta_{0}\right)}{Z\left(\beta\right)}.
\end{align*}
Hence 
\[
Ex=kT_{0}D\left(\left.P_{\beta}\right\Vert P_{\beta_{0}}\right).\]
This equation appeared already in \cite{Harremoes1993}.

\subsection{Portfolio theory}

The relation between information theory and gambling was
established by Kelly \cite{Kelly1956}. Logarithmic terms appear because we
are interested in the exponent in the exponential growth rate of our
wealth. Later Kelly's approach has been generalized to trading of stocks
although the relation to information theory is weaker \cite{Cover1991}.

Let $X_{1},X_{2},\dots,X_{k}$ denote \emph{price relatives} for
a list of $k$ assets. For instance $X_{5}=1.04$ means that asset
no. 5 increases its value by 4 \%. Such price relatives are mapped
into a price relative vector
$\vec{X}=\left(X_{1},X_{2},\dots,X_{k}\right).$
\begin{Example}
A special asset is the \emph{safe asset} where the price
relative is 1 for any possible price relative vector. Investing in this
asset corresponds to placing the money at a safe place with interest
rate equal to 0 \% .
\end{Example}
A \emph{portfolio} is a probability vector
$\vec{b}=\left(b_{1},b_{2},\dots,b_{k}\right)$
where for instance $b_{5}=0.3$ means that 30 \% of the money is
invested in asset no. 5. We note that a portfolio may be traded just like
the original assets. The price relative for the portfolio $\vec{b}$
is $X_{1}\cdot b_{1}+X_{2}\cdot b_{2}+\dots+X_{k}\cdot
b_{k}=\left\langle \vec{X},\vec{b}\right\rangle .$
The original assets may be considered as extreme points in the
set of portfolios. If an asset has the property that the price
relative is only positive for one of the possible price relative vectors,
then we may call it a \emph{gambling asset}. 

We now consider a situation where the assets are traded once
every day. For a sequence of price relative vectors
$\vec{X}_{1},\vec{X_{2}},\dots\vec{X}_{n}$
and \emph{a constant re-balancing portfolio} $\vec{b}$ the
wealth after $n$ days is
\begin{eqnarray}
S_{n} & = & \prod_{i=1}^{n}\left\langle
\vec{X}_{i},\vec{b}\right\rangle \\
& = & \exp\left(\sum_{i=1}^{n}\ln\left(\left\langle
\vec{X}_{i},\vec{b}\right\rangle \right)\right)\\
& = & \exp\left(n\cdot E\left[\ln\left\langle
\vec{X},\vec{b}\right\rangle \right]\right)
\end{eqnarray}
where the expectation is taken with respect to the empirical
distribution of the price relative vectors. Here $E\left[\ln\left\langle
\vec{X},\vec{b}\right\rangle \right]$
is proportional to the \emph{doubling rate} and is denoted
$W\left(\vec{b},P\right)$
where $P$ indicates the probability distribution of $\vec{X}$.
Our goal is to maximize $W\left(\vec{b},P\right)$ by choosing an
appropriate portfolio $\vec{b}.$

\begin{Definition}
Let $\vec{b}_{1}$ and $\vec{b}_{2}$ denote two portfolios. We
say
that $\vec{b}_{1}$ \emph{dominates} $\vec{b}_{2}$ if
$\left\langle \vec{X}_{j},\vec{b}_{1}\right\rangle
\geq\left\langle
\vec{X}_{j},\vec{b}_{2}\right\rangle $
for any possible price relative vector $\vec{X}_{j}$
$j=1,2,\dots,n.$
We say that $\vec{b}_{1}$ \emph{strictly dominates}
$\vec{b}_{2}$
if $\left\langle \vec{X}_{j},\vec{b}_{1}\right\rangle
>\left\langle \vec{X}_{j},\vec{b}_{2}\right\rangle $
for any possible price relative vector $\vec{X}_{j}$
$j=1,2,\dots,n.$ 
A set $A$ of assets is said to dominate the set of assets $B$ if any asset in $B$ is dominated by a portfolio of assets in $A.$
\end{Definition}

The maximal doubling rate does not change if dominated assets are removed. Sometimes assets that are dominated but not
strictly dominated may lead to non-uniqueness of the optimal portfolio.

Let $\vec{b}_{P}$ denote a portfolio that is optimal for $P$ and define 
\begin{equation}
G(P) =W\left(\vec{b}_{P},P\right).\label{eq:G}
\end{equation}
The regret of choosing a portfolio that is optimal for $Q$ when the
distribution is $P$ is given by the regret function
\begin{equation}
D_G (P,Q)=W\left(\vec{b}_{P},P\right)-W\left(\vec{b}_{Q},P\right).\label{eq:Bregman}
\end{equation}
If $\vec{b}_{Q}$ is not uniquely determined we take a minimum
over all $\vec{b}$ that are optimal for $Q.$

\begin{Example}\label{ex:port}
Assume that the price relative vector is $\left( 2,\nicefrac{1}{2}\right)$ with probability $1-t$ and $\left(\nicefrac{1}{2},2\right)$ with probability $t$. Then the portfolio concentrated on the first asset is optimal for $t\leq \nicefrac{1}{5}$ and the portfolio concentrated on the second asset is optimal for $t> \nicefrac{4}{5}$. For values of $t$ between $\nicefrac{1}{5}$ and $\nicefrac{4}{5}$ the optimal portfolio invests money on both assets as illustrated in Figure \ref{fig:port}.

\begin{figure}
\begin{center}
\begin{tikzpicture}[line width=1.2pt, line cap=round,line join=round,>=triangle 45,x=9.89423076923077cm,y=9.714879154078563cm,scale=0.7]

\draw[->,color=black] (0.,0.) -- (1.1177000323939097,0.);
\foreach \x in {,0.2,0.4,0.6,0.8,1,0}
\draw[shift={(\x,0)},color=black] (0pt,2pt) -- (0pt,-2pt);

\draw[->,color=black] (0.,0.) -- (0.,0.9124392614188542);
\foreach \y in {0.2,0.4,0.6,0.8}
\draw[shift={(0,\y)},color=black] (2pt,0pt) -- (-2pt,0pt);

\draw [red] plot [smooth] coordinates { (0.2,0.4182935787194957) (0.24350034788969488,0.36120950042115774) (0.2870003018694437,0.3168468268696912) (0.3305002558491924,0.2817583882332283) (0.36050022411108806,0.26258526897379575) (0.4010001812646472,0.242875609771641) (0.4475001320705855,0.2286661976310767) (0.4880000892241446,0.22343157468501623) (0.540500033682462,0.22642765345704483) (0.6079999622717273,0.24665640576588088) (0.6469999210121916,0.26700668167457386) (0.6844998813395612,0.29285980246773613) (0.7279998353193099,0.33105261588381335) (0.7519998099288264,0.35616328670391684) (0.7999997591478595,0.4158879744441841)};

\draw [dash pattern=on 4pt off 4pt] (0.2,0.4182935787194957)-- (0.2,0.);
\draw [dash pattern=on 4pt off 4pt] (0.8000000094532285,0.41588832144092486)-- (0.8,0.);

\draw (0.02926465824424991,0.8875607385811478) node[anchor=north west] {$G(1-t,t)$};
\draw (1.05,0.09) node[anchor=north west] {$t$};
\draw (-0.02,-0.01) node[anchor=north west] {$0$};
\draw (0.16,-0.01) node[anchor=north west] {$0.2$};
\draw (0.76,-0.01) node[anchor=north west] {$0.8$};
\draw (0.98,-0.01) node[anchor=north west] {$1$};

\draw [color=blue] (0.,0.6931471805599453)-- (0.19827171206890273,0.4182935787194957);
\draw [color=blue] (1.,0.6931471805599453)-- (0.8000000094532285,0.41588832144092486);

\draw [fill=black] (0.8000000094532285,0.41588832144092486) circle (2.0pt);
\draw [fill=black] (0.19827171206890273,0.4182935787194957) circle (2.0pt);
\draw [fill=black] (0.,0.6931471805599453) circle (2.0pt);
\draw [fill=black] (1.,0.6931471805599453) circle (2.0pt);

\end{tikzpicture}
\end{center}
\caption
{The function $G$ for the price relative vectors in Example \ref{ex:port}.\label{fig:port}}

\end{figure}
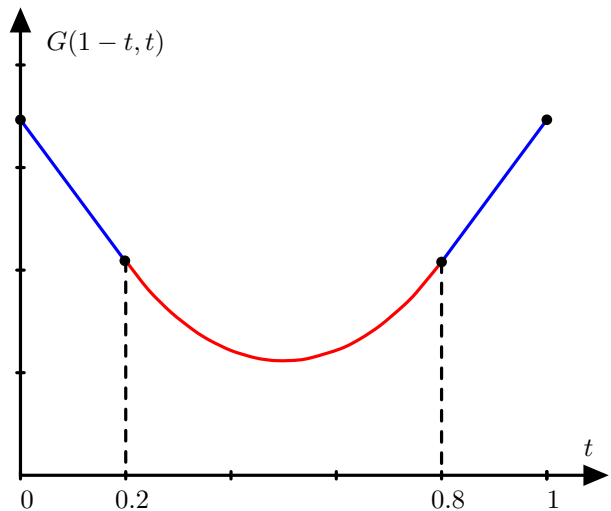
\end{Example}

\begin{lem}
If there are only two price relative vectors and the regret function is strict then either one of the assets dominates all other assets or two of the assets are orthogonal gambling assets that dominate all other assets.
\end{lem}
\begin{proof}
We will assume that no assets are dominated by other assets. Let
\begin{eqnarray*}
\vec{X} & = & \left(X_{1},X_{2},\dots,X_{k}\right)\\
\vec{Y} & = & \left(Y_{1},Y_{2},\dots,Y_{k}\right)
\end{eqnarray*}
denote the two price relative vectors. Without loss of
generality we may assume that 
\[
\frac{X_{1}}{Y_{1}}\geq\frac{X_{2}}{Y_{2}}\geq\dots\geq\frac{X_{k}}{Y_{k}}.
\]
If $\frac{X_{i}}{Y_{i}}=\frac{X_{i+1}}{Y_{i+1}}$ then
$\frac{X_{i}}{X_{i+1}}=\frac{Y_{i}}{Y_{i+1}}$
so that if $X_{i}\leq X_{i+1}$ then $Y_{i}\leq Y_{i+1}$ and the
asset $i$ is dominated by the asset $i+1.$ Since we have assumed that no assets are dominated we may assume that 
\[
\frac{X_{1}}{Y_{1}}>\frac{X_{2}}{Y_{2}}>\dots>\frac{X_{k}}{Y_{k}}.
\]
If $P=\left(1-t,t\right)$ is a probability vector over the two price relative vectors then according to \cite{Cover1991} the
portfolio $\vec{b}=\left(b_{1},b_{2},\dots,b_{n}\right)$ is optimal if and only if 
\[
(1-t)\frac{X_{i}}{b_{1}X_{1}+\dots+b_{k}X_{k}}+t\frac{Y_{i}}{b_{1}Y_{1}+\dots+b_{k}Y_{k}}\leq1
\]
for all $i\in\left\{ 1,2,\dots,k\right\} $ with equality if
$b_{i}>0.$
Assume that the portfolio $\vec{b}=\delta_{j}$ is optimal.
Now
\[
\left(1-t\right)\frac{X_{j+1}}{X_{j}}+t\frac{Y_{j+1}}{Y_{j}} \leq1
\]
is equivalent to
\begin{equation}\label{eq:hoejreendepunkt}
t \leq\frac{\frac{X_{j}}{Y_{j+1}}-\frac{X_{j+1}}{Y_{j+1}}}{\frac{X_{j}}{Y_{j}}-\frac{X_{j+1}}{Y_{j+1}}}.
\end{equation}
Similarly
\[
\left(1-t\right)\frac{X_{j-1}}{X_{j}}+t\frac{Y_{j-1}}{Y_{j}} \leq1
\]
is equivalent to 
\begin{equation}\label{eq:venstreendepunkt}
t \geq\frac{\frac{X_{j}}{Y_{j-1}}-\frac{X_{j-1}}{Y_{j-1}}}{\frac{X_{j}}{Y_{j}}-\frac{X_{j-1}}{Y_{j-1}}}.
\end{equation}
We have to check that 
\[
\frac{\frac{X_{j}}{Y_{j-1}}-\frac{X_{j-1}}{Y_{j-1}}}{\frac{X_{j}}{Y_{j}}-\frac{X_{j-1}}{Y_{j-1}}}<\frac{\frac{X_{j}}{Y_{j+1}}-\frac{X_{j+1}}{Y_{j+1}}}{\frac{X_{j}}{Y_{j}}-\frac{X_{j+1}}{Y_{j+1}}},
\]
which is equivalent with
\[
0 <X_{j}Y_{j-1}-Y_{j-1}X_{j+1}-Y_{j}X_{j-1}-\left(X_{j}Y_{j+1}-Y_{j+1}X_{j-1}-Y_{j}X_{j+1}\right).
\]
The right hand side equals the determinant
\[
\left|\begin{array}{cc}
X_{j+1}-X_{j-1} & X_{j}-X_{j-1}\\
Y_{j+1}-Y_{j-1} & Y_{j}-Y_{j-1}
\end{array}\right|,
\]
which is positive because asset $j$ is not dominated by a
portfolio
based on asset $j-1$ and asset $j+1.$ 

We see that the portfolio concentrated in asset $j$ is optimal
for
$t$ in an interval of positive length and the regret between distributions in such an interval will be zero. In particular the regret will not be strict.

Strictness of the regret function is only possible if there are only two assets and if a portfolio concentrated on one of these assets is only optimal for a singular probability measure. According to the formulas for the end points of intervals (\ref{eq:hoejreendepunkt}) and (\ref{eq:venstreendepunkt}) this is only possible if the assets are gambling assets.
\end{proof}

\begin{thm}\label{strict}
If the regret function is strict it equals information divergence, i.e.
\begin{equation}
D_G (P,Q)=D\left(P\Vert Q\right).\label{eq:perfekt}
\end{equation}
\end{thm}
\begin{proof}
If the regret function is strict then it is also strict when we restrict to two price relative vectors. Therefore any two price relative vectors are orthogonal gambling assets. If the assets are orthogonal gambling assets we get the type of
gambling described by Kelly \cite{Kelly1956}. For gambling equation can easily be derived \cite{Cover1991}.
\end{proof}
 
\section{Sufficiency Conditions}

In this section we will introduce some conditions on a
regret function. Under some mild conditions they turn out to be
equivalent.
\begin{thm}
Let $D_{F}$ denote a regret function based on a continuous and convex function $F$ defined on the state space
of a finite dimensional $C^{*}$-algebra. If the state space has at
least three orthogonal states then the following conditions are equivalent.
\begin{itemize}
\item The function $F$ equals entropy times a negative constant plus an affine function.
\item The regret $D_{F}$ is proportional to information divergence.
\item The regret is monotone.
\item The regret is satisfies sufficiency.
\item The regret is local.
\end{itemize}
\end{thm}

In the rest of this section we will describe each of these
equivalent conditions and prove that they are actually
equivalent. The theorems and proofs
will be stated so that they hold even for more general state
spaces than the ones considered in this paper.

\subsection{Entropy and Information Divergence}

\begin{Definition}
Let $s$ denote an element in a state space. The \emph{entropy}
of $s$ is be defined as 
\[
H\left(s\right)=\inf\left(-\sum_{i=1}^{n}\lambda_{i}\ln\left(\lambda_{i}\right)\right)
\]
where the infimum is taken over all decompositions
$s=\sum_{i=1}^{n}\lambda_{i}s_{i}$ of $s$ into pure states
$s_i$.
\end{Definition}
 This definition of the entropy of a state was first given by
Uhlmann \cite{Uhlmann1970}. Using that entropy is decreasing under
majorization we see that the entropy of $s$
is attained at an orthogonal decomposition \cite{Harremoes2016d}
and we obtain the familiar equation
\[
H(s) = - \mathrm{tr}\left[s\ln(s)\right] .
\]

In general this definition of entropy does not provide a concave
function on a convex set. For instance the entropy of points in the
square has local maximum in the four different points. A
characterization
of the convex sets with concave entropy functions is lacking.
\begin{Definition}
If the entropy is a concave function then the Bregman divergence
$D_{-H}$
is called \emph{information divergence}.
\end{Definition}
The information divergence is also called \emph{Kullback-Leibler
divergence},
\emph{relative entropy} or \emph{quantum relative entropy}. In a
C*-algebra we get
\begin{align*}
D_{-H}\left(s_1 ,s_2 \right) &
=-H\left(s_1 \right)-\left(-H\left(s_2 \right)+\left\langle
s_1 -s_2 ,-\nabla H\left(s_2 \right)\right\rangle \right)\\
& =H\left(s_2 \right)-H\left(s_1 \right)+\left\langle s_1 -s_2 ,\nabla
H\left(s_2 \right)\right\rangle \\
&
=\mathrm{tr}\left[f\left(s_2 \right)\right]-\mathrm{tr}\left[f\left(s_1 \right)\right]+\mathrm{tr}\left[\left(s_1 -s_2 \right)f'\left(s_2 \right)\right]\\
&
=\mathrm{tr}\left[f\left(s_2 \right)-f\left(s_1 \right)+\left(s_1 -s_2 \right)
f'\left(s_2 \right)\right]
\end{align*}
where $f\left(x\right)=-x\ln\left(x\right).$ Now
$f'\left(x\right)=-\ln\left(x\right)-1$
so that
\begin{align*}
f\left(s_2 \right)-f\left(s_1 \right)+\left(s_1 - s_2 \right)f'\left(s_2 \right)
&
=-s_2 \ln\left(s_2 \right)+s_1 \ln\left(s_1 \right)+\left(s_1 -s_2 \right)\left(-\ln\left(s_2 \right)-1\right)\\
 & =s_1 \left(\ln\left(s_1 \right)-\ln\left(s_2 \right)\right)+ s_2 -s_1.
\end{align*}
Hence 
\[
D_{-H}\left(s_1 ,s_2 \right)=\mathrm{tr}\left[s_1 \left(\ln\left(s_1 \right)-\ln\left(s_2 \right)\right)+s_2 - s_1 \right].
\]
For states $s_1, s_2$ it reduces to the well-known formula
\[
D_{-H}\left(s_1 ,s_2 \right)=\mathrm{tr}\left[s_1 \ln\left(s_1 \right)-s_1 \ln\left(s_2 \right)\right].
\]

\subsection{Monotonicity }

We consider a set $\mathcal{T}$ of maps
of the state space into itself. The set $\mathcal{T}$ will be used to
represent those transformations
that we are able to perform on the state space before we choose
a feasible action $a\in\mathcal{A}$. Let
$\Phi:\mathcal{S}\curvearrowright\mathcal{S}$
denote a map. Then the dual map $\Phi^*$ maps actions into actions and is given by
\[
\left\langle a,\Phi\left(s\right)\right\rangle  = 
\left\langle \Phi^* (a),s\right\rangle .
\]
\begin{Prop}[The principle of lost opportunities]\label{prop:lostoppotunities}
If $\Phi^*$ maps the set of feasible actions $\mathcal{A}$  into itself then 
\begin{equation}
F\left(\Phi\left(s\right)\right)\leq
F\left(s\right).\label{eq:aftagende}
\end{equation}
 
\end{Prop}
\begin{proof}
If $a\in\mathcal{A}$  then 
\begin{eqnarray*}
\left\langle a,\Phi\left(s\right)\right\rangle & = &
\left\langle \Phi^*(a),s\right\rangle \\
 & \leq & F\left(s\right)
\end{eqnarray*}
because $\Phi^*(a)\in\mathcal{A}$. Inequality
(\ref{eq:aftagende}) follows because
$F\left(\Phi\left(s\right)\right)=\sup_{a}\left\langle
a,\Phi\left(s\right)\right\rangle .$
\end{proof}

\begin{cor}[Semi-monotonicity]
Let $\Phi$ denote a map of the state space into itself such that $\Phi^*$ maps the set of feasible actions $\mathcal{A}$ into itself and let $s_{2}$
denote
a state that minimizes the function $F$. If $D_{F}$ is a Bregman divergence then 
\begin{equation}
D_{F}\left(\Phi\left(s_{1}\right),\Phi\left(s_{2}\right)\right)\leq
D_{F}\left(s_{1},s_{2}\right).\label{eq:aftagende-1}
\end{equation}
\end{cor}

\begin{proof}
Since $s_{2}$ minimizes $F$ and $F$ is differentiable we have
$\nabla F\left(s_{2}\right)=0$.
Since $s_{2}$ minimizes $F$ and
$F\left(\Phi\left(s_{2}\right)\right)\leq F\left(s_{2}\right)$
we also have that $\Phi\left(s_{2}\right)$ minimizes $F$ and
that
$\nabla F\left(\Phi\left(s_{2}\right)\right)=0$. Therefore 
\begin{align*}
D_{F}\left(\Phi\left(s_{1}\right),\Phi\left(s_{2}\right)\right)
&
=F\left(\Phi\left(s_{1}\right)\right)-\left(F\left(\Phi\left(s_{2}\right)\right)+\left\langle
\Phi\left(s_{1}\right)-\Phi\left(s_{2}\right),\nabla
F\left(\Phi\left(s_{2}\right)\right)\right\rangle \right)\\
&
=F\left(\Phi\left(s_{1}\right)\right)-F\left(\Phi\left(s_{2}\right)\right)\\
 & \leq F\left(s_{1}\right)-F\left(s_{2}\right)\\
 & =D_{F}\left(s_{1},s_{2}\right),
\end{align*}
which proves the inequality.
\end{proof}

Next we introduce the stronger notion of monotonicity. 
\begin{Definition}
Let $D_{F}$ denote a regret function on the state space $\mathcal{S}$ of a finite dimensional C*-algebra. Then $D_{F}$ is said to be \emph{monotone} if 
\[
D_{F}\left(\Phi\left(s_{1}\right),\Phi\left(s_{2}\right)\right)\leq
D_{F}\left(s_{1},s_{2}\right)
\]
for any affine map $\Phi:S\to S.$ 
\end{Definition}
\begin{Prop}\label{prop:monoBreg}
If a regret function $D_F$ based on a convex and continuous function $F$ is monotone then it is a Bregman divergence.
\end{Prop}
\begin{proof}
Assume that $D_F$ is monotone. We have to prove that $F$ is differentiable. Since $F$ is convex it is sufficient to prove that any restriction of $F$  to a line segment is differentiable. Let $s_0$ and $s_1$ denote states that are the end points of a line segment. The restriction of $F$ to the line segment is given by the convex and continuous function $f(t)= F((1-t)s_0 +ts_1 )$ so we have to prove that $f$ is differentiable.

If $0<t_1 <t_2 <1$ then according to Equation (\ref{eq:rightderiv}) we have
\[
D_{F}\left((1-t_2 )s_{0} +t_2 s_{1} ,(1-t_1 )s_{0} +t_1 s_{1} \right)=f\left(t_2 \right)-\left(f\left(t_1 \right)+\left(t_2 - t_1\right)\cdot f'_+ \left( t_1 \right) \right)
\]
where $f_+$ denotes the denote the derivative from the right. A dilation by a factor $r\leq 1$ around $s_0$ decreases the regret so that 
\begin{equation}
r\to f\left(r\cdot t_2 \right)-\left(f\left(r\cdot t_1 \right)+r\cdot\left(t_2 - t_1\right)\cdot f'_+ \left(r\cdot t_1 \right) \right) \label{eq:hoejreafledt}
\end{equation}
is increasing. Since $f$ is convex the function $r\to f'_+ \left(r\cdot t_1 \right)$ is increasing. Assume that $f$ is not differentiable so that $r\to f'_+ \left(r\cdot t_1 \right)$ has a positive jump as illustrated on Figure \ref{Fig:monobreg}. This contradicts that the function (\ref{eq:hoejreafledt}) is increasing. Therefore $f'_+$ is continuous and $f$ is differentiable.
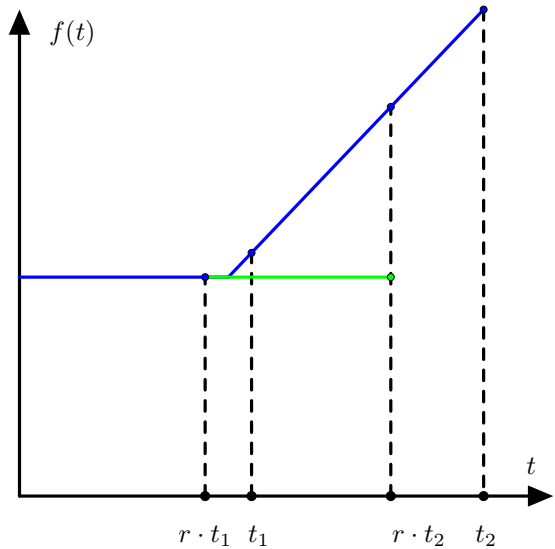
\begin{figure}
\begin{centering}
\begin{tikzpicture}[line width=1.2pt,line cap=round,line join=round,>=triangle 45,x=9.395604395604412cm,y=9.924107142857167cm,scale=0.65]
\draw[->,color=black] (0.,0.) -- (1.15,0.);
\draw[->,color=black] (0.,0.) -- (0.,1);

\draw[color=blue] (0,0.45) -- (0.4499987575911538,0.45);
\draw[color=blue] (0.4500007324771293,0.4500007324771293) -- (0.998624824324231,0.998624824324231);

\draw [dash pattern=on 4pt off 4pt] (1.,1.)-- (1.,0.);
\draw [dash pattern=on 4pt off 4pt] (0.50,0.)-- (0.50,0.50);
\draw [dash pattern=on 4pt off 4pt] (0.8,0.)-- (0.80,0.80);
\draw [dash pattern=on 4pt off 4pt] (0.4,0.)-- (0.4,0.45);

\draw [color=black](0.78,-0.04) node[anchor=north west] {$r\cdot t_2$};
\draw [color=black](0.96,-0.04) node[anchor=north west] {$t_2$};
\draw [color=black](0.32,-0.04) node[anchor=north west] {$r\cdot t_1$};
\draw [color=black](0.47,-0.04) node[anchor=north west] {$t_1$};
\draw (1.07,0.1) node[anchor=north west] {$t$};
\draw (0.04164903189233505,1) node[anchor=north west] {$f(t)$};
\draw [green](0.4,0.45)-- (0.8,0.45);

\draw [fill=black] (1.,0.) circle (2pt);
\draw [fill=black] (0.50,0.) circle (2pt);
\draw [fill=black] (0.8,0.) circle (2pt);
\draw [line width=0.3pt,fill=blue] (0.4,0.45) circle (2pt);
\draw [fill=black] (0.4,0.) circle (2pt);
\draw [line width=0.3pt,fill=blue] (0.50,0.50) circle (2pt);
\draw [line width=0.3pt,fill=blue] (0.80,0.80) circle (2pt);
\draw [line width=0.3pt,fill=blue] (1.,1.) circle (2pt);
\draw [line width=0.3pt,fill=green] (0.8,0.45) circle (2pt);

\end{tikzpicture}
\par\end{centering}
\caption{Example of a dilation that increases regret.}\label{Fig:monobreg}
\end{figure}
\end{proof}

Recently it has been proved that information divergence on a
complex Hilbert space is decreasing under positive trace preserving maps
\cite{Mueller-Hermes2016,Christandl2016}.
Previously this was only known to hold if some extra condition
like complete positivity or 2-positivity was assumed \cite{Petz2003}.
\begin{thm}
Information divergence is monotone under any positive trace
preserving map on the states of a finite dimensional $C^{*}$-algebra.
\end{thm}
\begin{proof}
Any finite dimensional $C^{*}$-algebra $\mathcal{B}$ can be
embedded in $\mathbb{B}\left(\mathbb{H}\right)$
and there exist a conditional expectation
$\mathbb{E}:\mathbb{B}\left(\mathbb{H}\right)\to\mathcal{B}.$
If $\Phi$ is a positive trace preserving map of the density
matrices of $\mathcal{B}$ into it self then $\Phi\circ\mathbb{E}$ is
positive and trace preserving on $\mathbb{B}\left(\mathbb{H}\right).$
According to M{\"u}ller-Hermes and Reeb \cite{Mueller-Hermes2016} we have\[
D\left(\left.\Phi\circ\mathbb{E}\left(s_{1}\right)\right\Vert
\Phi\circ\mathbb{E}\left(s_{2}\right)\right)\leq
D\left(\left.s_{1}\right\Vert
s_{2}\right)
\]
for density matrices in $\mathbb{B}\left(\mathbb{H}\right).$ In
particular this inequality holds for density matrices in $\mathcal{B}$ and for such
matrices we have $\mathbb{E}\left(s_{i}\right)=s_{i}$.
\end{proof}

\subsection{Sufficiency}

The notion of sufficiency plays an important role in statistics and related fields. We shall present a definition of sufficiency that is based on \cite{Petz1988}, but there are a number of other equivalent ways of defining this
concept. We refer to \cite{Jencova2006} where the notion of sufficiency
is discussed in great detail.
\begin{Definition}
Let $\left(s_{\theta}\right)_{\theta}$ denote a family of states and let $\Phi$ denote an affine map
$\mathcal{S}\to\mathcal{T}$
where $\mathcal{S}$ and $\mathcal{T}$ denote state spaces. 
A \emph{recovery map} is an affine map
$\Psi:\mathcal{T}\to\mathcal{S}$
such that $\Psi\left(\Phi\left(s_{\theta}\right)\right)=s_{\theta}.$
The map $\Phi$ is said to be \emph{sufficient} for
$\left(s_{\theta}\right)_{\theta}$
if $\Phi$ has a recovery map. 
\end{Definition}
\begin{Prop}
Assume  $D_{F}$ is a regret function based on a convex and continuous function $F$ and assume that $\Phi$ is sufficient for
$s_{1}$
and $s_{2}$ with recovery map $\Psi$. Assume that both $\Phi^*$ and $\Psi^*$ map the set of feasible actions $\mathcal{A}$ into itself. Then
\[
D_{F}\left(\Phi\left(s_{1}\right),\Phi\left(s_{2}\right)\right)=D_{F}\left(s_{1},s_{2}\right).
\]
\end{Prop}
\begin{proof}
According to the principle of lest opportunities (Proposition \ref{prop:lostoppotunities}) we have
\begin{align*}
F\left(s_{2}\right) &
=F\left(\Psi\left(\Phi\left(s_{2}\right)\right)\right)\\
 & \leq F\left(\Phi\left(s_{2}\right)\right)\\
 & \leq F\left(s_{2}\right).
\end{align*}
Therefore
$F\left(\Phi\left(s_{2}\right)\right)=F\left(s_{2}\right).$
Let $a$ denote an action that is optimal for $s_{2}.$ Then
\begin{align*}
F\left(\Phi\left(s_{2}\right)\right) & =F\left(s_{2}\right)\\
 & =\left\langle a,s_{2}\right\rangle \\
& =\left\langle
a,\Psi\left(\Phi\left(s_{2}\right)\right)\right\rangle \\
& =\left\langle \Psi^*(a),\Phi\left(s_{2}\right)\right\rangle\end{align*}
and we see that $\Psi^*(a)$ is optimal for
$\Phi\left(s_{2}\right).$
Now
\begin{align*}
D_{F}\left(s_{1},s_{2}\right) &
=\inf_{a}\left(F\left(s_{1}\right)-\left\langle
a,s_{1}\right\rangle \right)\\
& =\inf_{a}\left(F\left(s_{1}\right)-\left\langle
\Psi^*(a),\Phi\left(s_{1}\right)\right\rangle \right)
\end{align*}
where the infimum is taken over actions $a$ that are optimal for $s_{2}.$ Then 
\begin{align*}
\inf_{a}\left(F\left(s_{1}\right)-\left\langle
\Psi^*(a),\Phi\left(s_{1}\right)\right\rangle \right) &
\geq\inf_{\tilde{a}}\left(F\left(\Phi\left(s_{1}\right)\right)-\left\langle
\tilde{a},\Phi\left(s_{1}\right)\right\rangle \right)\\
&
=D_{F}\left(\Phi\left(s_{1}\right),\Phi\left(s_{2}\right)\right)
\end{align*}
so we have $D_{F}\left(s_{1},s_{2}\right)\geq
D_{F}\left(\Phi\left(s_{1}\right),\Phi\left(s_{2}\right)\right).$
The reverse inequality is proved in the same way.
\end{proof}
The notion of sufficiency as a property of divergences was
introduced in \cite{Harremoes2007a}. The crucial idea of restricting the
attention to maps of the state space into itself was introduced
in \cite{Jiao2014}. It was shown in \cite{Jiao2014} that a Bregman
divergence on the simplex of distributions on an alphabet that
is not binary and satisfies sufficiency equals information divergence up a multiplicative
factor. Here we extend the notion of sufficiency from Bregman
divergences to regret functions.
\begin{Definition}
Let $D_F$ denote a regret function based on a convex and continuous function $F$ on a state space ${\mathcal S}$. We say $D_{F}$ satisfies \emph{sufficiency} if 
\[
D_{F}\left(\Phi\left(s_{1}\right),\Phi\left(s_{2}\right)\right)=D_{F}\left(s_{1},s_{2}\right)
\]
for any affine map $\mathcal{S}\to\mathcal{S}$ that
is sufficient for $\left(s_{1},s_{2}\right).$ 
\end{Definition}

\begin{Prop}
\label{prop:sufficiency}Let $D_F$ denote a regret function based on a convex and continuous function $F$ on a state space ${\mathcal S}$. If the regret function $D_{F}$ is monotone then it satisfies sufficiency.
\end{Prop}
\begin{proof}
Assume that the regret function $D_{F}$ is monotone. Let $s_{1}$ and $s_{2}$ denote two states and let $\Phi$ and $\Psi$ denote maps on the state space such that
$\Phi\left(\Psi\left(s_{i}\right)\right)=s_{i},\,i=1,2$~.
Then 
\begin{align*}
D_{F}\left(s_{1},s_{2}\right) &
=D_{F}\left(\Phi\left(\Psi\left(s_{1}\right)\right),\Phi\left(\Psi\left(s_{2}\right)\right)\right)\\
& \leq
D_{F}\left(\Psi\left(s_{1}\right),\Psi\left(s_{2}\right)\right)\\
 & \leq D_{F}\left(s_{1},s_{2}\right).
\end{align*}
Hence
$D_{F}\left(\Psi\left(s_{1}\right),\Psi\left(s_{2}\right)\right)=D_{F}\left(s_{1},s_{2}\right).$
\end{proof}
Combining the previous results we get that information
divergence satisfies sufficiency. Under some conditions
there exists an inverse version of Proposition
\ref{prop:sufficiency}
stating that if monotonicity holds with equality then the
map is sufficient. In statistics where the state space is a simplex,
this result is well established. For density matrices over the
complex numbers it has been proved for completely positive maps in
\cite{Jencova2006}. Some new results on this topic can be found in
\cite{Jencova2017}.

\subsection{Locallity}

Often it is relevant to use the following weak version of the
sufficiency
property.
\begin{Definition}
Let $D_F$ denote a regret function based on a convex and continuous function $F$ on a state space ${\mathcal S}$. The regret function $D_{F}$ is said to be local if 
\[
D_{F}\left(s_{1},t\cdot s_{1}+\left(1-t\right)\cdot \sigma\right)=D_{F}\left(s_{1},t\cdot s_{1}+\left(1-t\right)\cdot \rho\right)
\]
when the states $\sigma$ and $\rho$ are orthogonal to
$s_{1}$
and $t\in\left]0,1\right[.$
\end{Definition}
\begin{Example}
On a 1-dimensional simplex (an interval) or on the Block sphere
any regret function $D_{F}$ is local.
The reason is that if $\sigma$ and $\rho$ are states that are
orthogonal
to $s_{1}$ then $\sigma=\rho.$
\end{Example}
\begin{Prop}Let $D_F$ denote a regret function based on a convex and continuous function $F$ on a state space ${\mathcal S}$. If the regret function $D_{F}$ satisfies sufficiency then $D_{F}$ is local. 
\end{Prop}
\begin{proof}
Let $\sigma$ and $\rho$ be states that are orthogonal to
$s_{1}.$ Let $p$ denote the projection supporting the state $s_0$.
 Let the maps $\Phi$ and $\Psi$ be defined by
\begin{align*}
\Phi\left(s\right) &
={\rm tr}(ps)\cdot s_{1}+(1-{\rm tr}(ps))\cdot \rho ,\\
\Psi\left(s\right) &
={\rm tr}(ps)\cdot s_{1}+(1-{\rm tr}(ps))\cdot \sigma.
\end{align*}
Then $\Phi\left(s_{1}\right)=\Psi\left(s_{1}\right)=s_{1}$ and
$\Phi\left(\sigma\right)=\rho$
and $\Psi\left(\rho\right)=\sigma.$ Therefore 
\begin{align*}
\Phi\left(t\cdot s_{1}+\left(1-t\right)\cdot \sigma\right) &
=t\cdot s_{1}+\left(1-t\right)\cdot \rho\\
\Psi\left(t\cdot s_{1}+\left(1-t\right)\cdot \rho\right) &
=t\cdot s_{1}+\left(1-t\right)\cdot \sigma
\end{align*}
 and 
\[
D_{F}\left(s_{1},t\cdot s_{1}+\left(1-t\right)\cdot \sigma\right)=D_{F}\left(s_{1},t\cdot s_{1}+\left(1-t\right)\cdot \rho\right).
\]
\end{proof}

\begin{thm}
Let $\mathcal{S}$ be the state space of a $C^*$-algebra with at least three orthogonal
states, and let $D_F$ denote a regret function based on a convex and continuous function $F$ on the state space ${\mathcal S}$.
If the regret function $D_{F}$ is local then it is the Bregman divergence generated by the entropy times a negative constant. 
\end{thm}
\begin{proof}
In the following proof we will assume that the regret function
is based on the convex function $F:\mathcal{S}\to\mathbb{R}.$  First we will prove that the regret function is a Bregman divergence.

Let $K$ denote the convex hull of a set $s_{0},s_{1},\dots
s_{n}$ of orthogonal states. For $x\in[0,1[$ let $g_{i}$ denote the function
$g_{i}\left(x\right)=D_{F}\left(s_{i},xs_{i}+\left(1-x\right)s_{i+1}\right)$.
Note that $g_{i}$ is decreasing and continuous from the left.
Let $P=\sum p_{i}s_{i}$ and $Q=\sum q_{i}s_{i}$ where $p_{i} , q_{i} \in\left] 0,1 \right[$ for all $i=0,1,2,\dots n$. If $F$ is
differentiable in $P$ then locality implies that 
\begin{eqnarray*}
D_{F}\left(P,Q\right) & = & \sum
p_{i}D_{F}\left(s_{i},Q\right)-\sum
p_{i}D_{F}\left(s_{i},P\right)\\
& = & \sum p_{i}g_{i}\left(q_{i}\right)-\sum
p_{i}g_{i}\left(p_{i}\right)\\
& = & \sum
p_{i}\left(g_{i}\left(q_{i}\right)-g_{i}\left(p_{i}\right)\right).
\end{eqnarray*}
Note that $P\to D_{F}\left(P,Q\right)$ is a convex function and
thereby it is continuous. Assume that $P_{0}$ is an arbitrary element in $K$ and let $\left(P_{n}\right)_{n\in\mathbb{N}}$ denote a sequence such that $P_{n}\to P_{0}$ for $n\to\infty.$ The sequence
$\left(P_{n}\right)_{n\in\mathbb{N}}$
can be choosen so that regret is differentiable in $P_{n}$ for
all $n\in\mathbb{N}.$ Further the sequence $P_{n}$ can be chosen
such that $p_{n,i}$ is increasing for all $i\neq j.$ Then 
\[
D_{F}\left(P_{0},Q\right)=\sum
p_{0,i}\left(g_{i}\left(q_{i}\right)-g_{i}\left(p_{0,i}\right)\right)+p_{0,j}g_{j}\left(p_{0,j}\right)-p_{0,j}\lim_{n\to\infty}g_{j}\left(p_{n,j}\right).
\]
Similarly, if the sequence $P_{n}$ can be chosen such that
$p_{n,i}$ is increasing for all $i\neq j,j+1$ then 
\begin{multline*}
D_{F}\left(P_{0},Q\right)=
\sum
p_{0,i}\left(g_{i}\left(q_{i}\right)-g_{i}\left(p_{0,i}\right)\right)+p_{0,j}g_{j}\left(p_{0,j}\right)-p_{0,j}\lim_{n\to\infty}g_{j}\left(p_{n,j}\right)\\
+p_{0,j+1}g_{j+1}\left(p_{0,j+1}\right)-p_{0,j+1}\lim_{n\to\infty}g_{j+1}\left(p_{n,j+1}\right),
\end{multline*}
which implies that
$p_{0,j+1}g_{j+1}\left(p_{0,j+1}\right)-p_{0,j+1}\lim_{n\to\infty}g_{j+1}\left(p_{n,j+1}\right)=0$
and that 
\[
\lim_{n\to\infty}g_{j+1}\left(p_{n,j+1}\right)=g_{j+1}\left(p_{0,j+1}\right)
\]
for all $j$. Therefore
\begin{equation}\label{eq:local}
D_{F}\left(P_{0},Q\right)=\sum
p_{0,i}\left(g_{i}\left(q_{i}\right)-g_{i}\left(p_{0,i}\right)\right)
\end{equation}
for all $P_{0},Q$ in the interior of $K$. In the following calculations we will assume that the distributions lie in the interior of $K$. The validity of the Bregman identity (\ref{eq:Bregmanid}) follows directly from Equation \ref{eq:local} implying that $D_F$ is a Bregman divergence.

As a function of $Q$ the regret is minimal when $Q=P.$ In the
following calculations we write $x=p_i$, $z=p_j$, $y=q_i$, and
$w=q_j$. If
$p_{\ell}=q_{\ell}$ for $\ell \neq i,j$ then non-negativity of
regret can be written as
\[
x\left(g_{i}\left(y\right)-g_{i}\left(x\right)\right)+z\left(g_{j}\left(w\right)-g_{j}\left(z\right)\right)\geq0
\]
and we note that this inequality should hold as long as
$x+z=y+w\leq 1.$ Permutation of $i$ and $j$ leads to the
inequality
\[
x\left(g_{j}\left(y\right)-g_{j}\left(x\right)\right)+z\left(g_{i}\left(w\right)-g_{i}\left(z\right)\right)\geq0
\]
that implies
\begin{equation}
x\left(g_{ij}\left(y\right)-g_{ij}\left(x\right)\right)+z\left(g_{ij}\left(w\right)-g_{ij}\left(z\right)\right)\geq0 \label{eq:sym}
\end{equation}
where $g_{ij}=\frac{g_{i}+g_{j}}{2}.$ 

Assume that $x=z=\frac{y+w}{2}$ in Inequality (\ref{eq:sym}). Then
\begin{eqnarray*}
x\left(g_{ij}\left(y\right)-g_{ij}\left(x\right)\right)+x\left(g_{ij}\left(w\right)-g_{ij}\left(x\right)\right)
& \geq & 0\\
g_{ij}\left(y\right)-g_{ij}\left(x\right)+g_{ij}\left(w\right)-g_{ij}\left(x\right)
& \geq & 0\\
\frac{g_{ij}\left(y\right)+g_{ij}\left(w\right)}{2} & \geq &
g_{ij}\left(x\right)
\end{eqnarray*}
so that $g_{ij}$ is mid-point convex, which for a measurable
function implies convexity. Therefore $g_{ij}$ is differentiable from
left and right. 

If $y=w$ and $x=y+\epsilon$ and $z=y-\epsilon$ then
we have
\[
\left(y+\epsilon\right)\left(g_{ij}\left(y\right)-g_{ij}\left(y+\epsilon\right)\right)+\left(y-\epsilon\right)\left(g_{ij}\left(y\right)-g_{ij}\left(y-\epsilon\right)\right)\geq0
\]
with equality when $\epsilon=0.$ We differentiate with respect
to
$\epsilon$ from right. 
\[
\left(g_{ij}\left(y\right)-g_{ij}\left(y+\epsilon\right)\right)+\left(y+\epsilon\right)\left(-g_{ij+}'\left(y+\epsilon\right)\right)-\left(g_{ij}\left(y\right)-g_{ij}\left(y-\epsilon\right)\right)+\left(y-\epsilon\right)\left(g_{ij-}'\left(y-\epsilon\right)\right)
\]
 which is positive for $\epsilon=0$ so that 
\begin{eqnarray}
-y\cdot g_{ij+}'\left(y\right)+y\cdot g_{ij-}'\left(y\right) &
\geq & 0\\
y\cdot g_{ij-}'\left(y\right) & \geq & y\cdot
g_{ij+}'\left(y\right).\label{eq:f_ij}
\end{eqnarray}
Since $g_{ij}$ is convex we have $g_{ij-}'\left(y\right)\leq
g_{ij+}'\left(y\right)$
which in combination Inequality (\ref{eq:f_ij}) implies that
$g_{ij-}'\left(y\right)=g_{ij+}'\left(y\right)$
so that $g_{ij}$ is differentiable. Since
$g_{i}=g_{ij}+g_{ik}-g_{jk}$
the function $g_{i}$ is also differentiable. 

As a function of $Q$ the Bregman divergence $D_{F}(P,Q)$ has a
minimum at $Q=P$ under the condition $\sum q_i =1$. Since the
functions $g_i$ are
differentiable we can characterize this minimum using Lagrange
multipliers. We have
\[
\frac{\partial}{\partial
q_{i}}D_{F}\left(P,Q\right)=p_{i}g_{i}'\left(q_{i}\right)
\]
 and 
\[
\frac{\partial}{\partial q_{i}}D_{F}\left(P,Q\right)_{\mid
Q=P}=p_{i}\cdot g_{i}'\left(p_{i}\right).
\]
Further $\frac{\partial}{\partial q_i}\sum q_i =1$ so there
exist a constant $c_{K}$ such that $p_{i}\cdot
g_{i}'\left(p_{i}\right)=c_{K}.$
Hence $g_{i}'\left(p_{i}\right)=\frac{c_{K}}{p_{i}}$ so that
$g_{i}\left(p_{i}\right)=c_{K}\cdot\ln\left(p_{i}\right)+m_{i}$
for some constant $m_{i}.$ 

Now we get 
\begin{eqnarray*}
D_{F}\left(P,Q\right) & = & \sum
p_{i}\left(g_{i}\left(q_{i}\right)-g_{i}\left(p_{i}\right)\right)\\
& = & \sum
p_{i}\left(\left(c_{K}\cdot\ln\left(q_{i}\right)+m_{i}\right)-\left(c_{K}\cdot\ln\left(p_{i}\right)+m_{i}\right)\right)\\
&=&-c_{K}\cdot\sum
p_{i}\ln\frac{p_{i}}{q_{i}}\\
 & = &D_{ -c_{K}\cdot H}\left(P,Q\right).
\end{eqnarray*}
Therefore there exists an affine function defined on $K$ such that
\begin{equation}\label{eq:affinentropi}
F_{\mid K}(P)=-c_{K}\cdot H_{\mid K}(P)+g_{K}
\end{equation}
for all $P$ in the interior of $K$. Since $H_K$ is continuous on $K$ Equation (\ref{eq:affinentropi}) holds for any $P\in K$. If each of the sets $K$ and $L$ is a 
simplex and $x\in K\cap L$ then  
\[
-c_{K}\cdot H_{\mid
K}\left(x\right)+g_{K}\left(x\right)=-c_{L}\cdot H_{\mid
L}\left(x\right)+g_{L}\left(x\right)
\]
 so that 
\[
\left(c_{L}-c_{K}\right)\cdot H_{\mid
K}\left(x\right)=g_{L}\left(x\right)-g_{K}\left(x\right).
\]
If $K\cap L$ has dimension greater than zero then the right hand side is affine so the left hand side is affine, which is only
possible when $c_{K}=c_{L}.$ Therefore we also have
$g_{L}\left(x\right)=g_{K}\left(x\right)$
for all $x\in K\cap L.$ Therefore the functions $g_{K}$ can be
extended to a single affine function on the whole of $\mathcal{S}.$ 
\end{proof}

\section{Applications}

\subsection{Information theory}

If only integer values of a code-length function $\ell$ are
allowed then there are only finitely many actions that are not
dominated. Therefore the function $F$ given by 
\[
F\left(P\right)=-\min_{\ell}\sum\ell\left(a\right)\cdot p_{a}
\]
is piece-wise linear. In particular $F$ is not differentiable so that the regret is not a Bregman divergence and cannot be monotone according to Proposition \ref{prop:monoBreg}. In information theory monotonicity of a divergence
function is closely related to the \emph{data processing inequality} and since
the data processing inequality is one of the most important tools for
deriving inequalities in information theory we need to modify our notion
of code-length function in order to achieve a data processing
inequality.

We now formulate a version of Kraft's inequality that allow the
code length function to be non-integer valued. 
\begin{thm}
\label{Theorem:Kraft}Let $\ell:\mathbb{A}\rightarrow\mathbb{R}$
be a function. Then the function $\ell$ satisfies Kraft's inequality
(\ref{eq:Kraft})
if and only if for all $\varepsilon>0$ there exists an integer
$n$ and a uniquely decodable fixed-to-variable length block code
$\kappa:\mathbb{A}^{n}\rightarrow\mathbb{B}^{\ast}$
such that 
\[
\left\vert
\bar{\ell}_{\kappa}\left(a^{n}\right)-\frac{1}{n}\sum_{i=1}^{n}\ell\left(a_{i}\right)\right\vert
\leq\varepsilon
\]
where $\bar{\ell}_{\kappa}\left(a^{n}\right)$ denotes the length $\ell_{\kappa}\left(a^{n}\right)$ divided by $n.$ The uniquely decodable block code can be chosen to be prefix free. 
\end{thm}
\begin{proof}
Assume that $\ell$ satisfies Kraft's inequality. Then 
\[
\sum_{a_{1}a_{2}...a_{n}\in\mathbb{A}^{n}}\beta^{\textrm{-}\sum_{i=1}^{n}\ell\left(a_{i}\right)}=\left(\sum_{a\in\mathbb{A}}\beta^{\textrm{-}\ell\left(a\right)}\right)^{n}\leq1^{n}=1.
\]
Therefore the function
$\tilde{\ell}:\mathbb{A}^{n}\rightarrow\mathbb{N}$
given by 
\[
\tilde{\ell}\left(a_{1}a_{2}...a_{n}\right)=\left\lceil
\sum_{i=1}^{n}\ell\left(a_{i}\right)\right\rceil
\]
is integer valued and satisfies Kraft's inequality (\ref{eq:Kraft}) and there
exists a prefix-free code $\kappa:\mathbb{A}^{n}\rightarrow\left\{
0,1\right\} ^{\ast}$ such that $\ell_{\kappa}\left(a_{1}a_{2}...a_{n}\right)=\tilde{\ell}\left(a_{1}a_{2}...a_{n}\right).$
Therefore 
\[
\left\vert
\bar{\ell}_{\kappa}\left(a_{1}a_{2}...a_{n}\right)-\frac{1}{n}\sum_{i=1}^{n}\ell\left(a_{i}\right)\right\vert
=\frac{1}{n}\left\vert
\left\lceil \sum_{i=1}^{n}\ell\left(a_{i}\right)\right\rceil
-\sum_{i=1}^{n}\ell\left(a_{i}\right)\right\vert \leq\frac{1}{n}
\]
so for any $\varepsilon>0$ choose $n$ such that
$\nicefrac{1}{n}\leq\varepsilon.$

Assume that for all $\varepsilon>0$ there exists a uniquely
decodable fixed-to-variable length code
$\kappa:\mathbb{A}^{n}\rightarrow\left\{ 0,1\right\} ^{\ast}$
such that 
\[
\left\vert
\bar{\ell}_{\kappa}\left(a_{1}a_{2}...a_{n}\right)-\frac{1}{n}\sum_{i=1}^{n}\ell\left(a_{i}\right)\right\vert
\leq\varepsilon
\]
for all strings $a_{1}a_{2}...a_{n}\in\mathbb{A}^{n}.$ Then
$n\bar{\ell}_{\kappa}\left(a_{1}a_{2}...a_{n}\right)$
satisfies Kraft's Inequality(\ref{eq:Kraft}) and 
\begin{align*}
\left(\sum_{a\in\mathbb{A}}\beta^{\textrm{-}\ell\left(a\right)}\right)^{n}
&
=\sum_{a_{1}a_{2}...a_{n}\in\mathbb{A}^{n}}\beta^{\textrm{-}\sum_{i=1}^{n}\ell\left(a_{i}\right)}\\
&
\leq\sum_{a_{1}a_{2}...a_{n}\in\mathbb{A}^{n}}\beta^{\textrm{-}n\left(\bar{\ell}_{\kappa}\left(a_{1}a_{2}...a_{n}\right)-\varepsilon\right)}\\
&
=\beta^{n\varepsilon}\sum_{a_{1}a_{2}...a_{n}\in\mathbb{A}^{n}}\beta^{\textrm{-}n\bar{\ell}_{\kappa}\left(a_{1}a_{2}...a_{n}\right)}\\
 & \leq\beta^{n\varepsilon}.
\end{align*}
Therefore
$\sum_{a\in\mathbb{A}}\beta^{\textrm{-}\ell\left(a\right)}\leq\beta^{\varepsilon}$
for all $\varepsilon>0$ and the result is obtained.
\end{proof}
Like in Bayesian statistics we focus on finite sequences.
Contrary to Bayesian statistics we should always consider a finite
sequence as a prefix of \emph{longer finite} sequences. Contrary to
frequential statistics we do not have to consider a finite sequence as a
prefix of an \emph{infinite} sequence.

If we minimize the mean code-length over functions that satisfy
Kraft's inequality (\ref{eq:Kraft}), but without an integer constraint the code-length
should be $\ell\left(a\right)=-\log_{\beta}\left(p_{a}\right)$ and the
function $F$ is given by
\[
F\left(P\right)=\sum_{a}p_{a}\cdot\log_{\beta}\left(p_{a}\right).
\]
The function $F$ is proportional to the Shannon entropy and the (negative)
proportionality factor is determined by the size of the output alphabet.

In lossy source coding and rate distortion theory it is important to choose a distortion function with tractable properties. The notion of sufficiency for divergence functions was introduced in \cite{Harremoes2007a} in order to characterize such tractable distortions functions. In this paper the main result was that sufficiency together with properties related to Bregman divergence lead directly to the information bottleneck method introduced by N. Tishby \cite{Tishby1999}. Logarithmic loss has also been studied for lossy compression in \cite{No2015}.

\subsection{Statistics}

In statistics one is often interested in scoring rules that are
local, which means a scoring rule where the payoff only depends on the
probability of the observed value and not on the predicted distribution over
unobserved values. The notion of locality has recently been extended by
Dawid, Lauritzen and Parry \cite{Dawid2012}, but here we shall focus on the original definition. The basic result is that the only
local strictly proper scoring rule is logarithmic score that was proved by Bernardo under the assumption that scoring rule is given by a smooth function \cite{Bernardo1978}.
\begin{Definition}
A \emph{local strictly proper scoring rule} is a scoring rule of
the form $f\left(x,Q\right)=g\left(Q\left(x\right)\right).$
\end{Definition}
\begin{thm}
On a finite space a local strictly proper scoring rule is given by a local regret function.
\end{thm}
\begin{proof}
The regret function of a local strictly proper scoring rule is
given by 
\[
D\left(P,Q\right)=\sum_{x}P\left(x\right)\left(g\left(P\left(x\right)\right)-g\left(Q\left(x\right)\right)\right).
\]
If $Q=\left(1-t\right)P+tQ_{i}$ and $P$ and $Q$ are mutually
singular then
\begin{align*}
D\left(P,Q\right) &
=\sum_{x}P\left(x\right)\left(g\left(P\left(x\right)\right)-g\left(\left(1-t\right)P\left(x\right)+tQ_{i}\left(x\right)\right)\right)\\
&
=\sum_{x}P\left(x\right)\left(g\left(P\left(x\right)\right)-g\left(\left(1-t\right)P\left(x\right)+0\right)\right)
\end{align*}
and we see that the regret does not depend on $Q_{i}$ because
$Q_{i}$
vanish on the support of $P.$ Therefore the regret function is
local.
\end{proof}
\begin{cor}
On a finite space with at least three elements a  local strictly proper
scoring rule is given by a function $g$ of the form
$g\left(x\right)=a\cdot\ln\left(x\right)+b$
for some constants $a$ and $b.$
\end{cor}
Also the notion of sufficiency plays an important role in
statistics. Here we will restrict the discussion to 1-dimensional
exponential families. A natural exponential family is a family of
probability distributions of the form
\[
\frac{\mathrm{d}P_{\beta}}{\mathrm{d}Q}=\frac{\exp\left(\beta
x\right)}{Z\left(\beta\right)}
\]
where $Q$ is a reference measure on the real numbers and $Z$ is
the moment generating function given by
$Z\left(\beta\right)=\int\exp\left(\beta
x\right)\,\mathrm{d}Qx$.
Then $x_{1}^{n}\to x_{1}+x_{2}+\dots+x_{n}$ is a sufficient
statistic for the family $\left(P_{\beta}^{\otimes n}\right)_{\beta}.$ 
\begin{Example}
In a Bernoulli model a sequence $x_{1}^{n}\in\left\{ 0,1\right\}
^{n}$ is predicted with probability 
\[
\prod_{i=1}^{n}p^{x_{i}}\left(1-p\right)^{1-x_{i}}=\exp\left(\left(\sum_{i=1}^{n}x_{1}\right)\ln\left(p\right)+\left(n-\sum_{i=1}^{n}x_{1}\right)\ln\left(1-p\right)\right).
\]
The function $x_{1}^{n}\to x_{1}+x_{2}+\dots+x_{n}$ induces a
sufficient map $\Phi$ from probability distributions on $\left\{
0,1\right\} ^{n}$
to probability distributions on $\left\{ 0,1,2,\dots,n\right\}
.$
The reverse map maps a measure concentrated in
$j\in\left\{ 0,1,2,\dots,n\right\} $
into a uniform distributions over sequences $x_{1}^{n}\in\left\{
0,1\right\} ^{n}$
that satisfy $\sum_{i=1}^{n}x_{1}=j.$
\end{Example}
The mean value of $P_{\beta}$ is 
\[
\int x\cdot\frac{\exp\left(\beta
x\right)}{Z\left(\beta\right)}\,\mathrm{d}Qx\,.
\]
The set of possible mean values is called the mean value range
and is an interval. Let $P^{\mu}$ denote the element in the
exponential family with mean value $\mu.$ Then a Bregman divergence on the
mean value range is defined by
$D\left(\lambda,\mu\right)=D\left(P^{\lambda}\left\Vert
P^{\mu}\right.\right).$
Note that the mapping $\mu\to P^{\mu}$ is not affine so the
Bregman divergence $D\left(\lambda,\mu\right)$ will in general not be
given by the formula for information divergence with the family of
binomial distributions as the only exception. Nevertheless the Bregman
divergence
$D\left(\lambda,\mu\right)$ encode important information about
the exponential family. In statistics it is common to use squared
Euclidean distance as distortion measure, but often it is better to use
the Bregman divergence $D\left(\lambda,\mu\right)$ as distortion
measure. Note that $D\left(\lambda,\mu\right)$ is only proportional to
squared Euclidean distance for the Gaussian location family.
\begin{Example}
An exponential distribution has density
\[
f_{\lambda}\left(x\right)=\begin{cases}
\frac{1}{\lambda}\exp\left(-\frac{x}{\lambda}\right)\,, &
\mathrm{for}\,x\geq0;\\
0\,, & else.
\end{cases}
\]
This leads to a Bregman divergence on the interval
$\left[0;\infty\right[$
given by
\begin{align*}
\int_{0}^{\infty}f_{\lambda}\left(x\right)\ln\left(\frac{f_{\lambda}\left(x\right)}{f_{\mu}\left(x\right)}\right)\,\mathrm{d}x
&
=\frac{\lambda}{\mu}-1-\ln\left(\frac{\lambda}{\mu}\right)\\
 & =D_{-\ln}\left(\lambda,\mu\right)
\end{align*}
This Bregman divergence is called the \emph{Isakura-Saito
distance}. The Isakura-Saito distance is defined on an unbounded set so our
previous results cannot be applied. Affine bijections on
$\left[0;\infty\right[$ have the form $\Phi\left(x\right)=k\cdot x$ for some constant
$k>0$.
The Isakura-Saito distance obviously satisfy sufficiency for
such maps and it is a simple exercise to check that the Isakura-Saito distance is
the only Bregman divergence on $[0,\infty ]$ that satisfies sufficiency. Any affine map
$\left[0;\infty\right[\to\left[0;\infty\right[$ is composed of a map $x\to k\cdot x$ where $k\geq0$ and a right
translation $x\to x+t$ where $t\geq0.$ The Itakura-Saito distance decreases
under right translations because
\begin{align*}
\frac{\partial}{\partial t}D_{-\ln}\left(\lambda+t,\mu+t\right)
& =\frac{\partial}{\partial
t}\left(\frac{\lambda+t}{\mu+t}-1-\ln\left(\frac{\lambda+t}{\mu+t}\right)\right)\\
&
=\frac{\left(\mu+t\right)-\left(\lambda+t\right)}{\left(\mu+t\right)^{2}}-\frac{1}{\lambda+t}+\frac{1}{\mu+t}\\
&
=-\frac{\left(\lambda-\mu\right)^{2}}{\left(\lambda+t\right)\left(\mu+t\right)^{2}}\leq0.
\end{align*}
Thus the Isakura-Saito distance is monotone.
\end{Example}

Both sufficiency and the Bregman identity are closely related to inference rules. In \cite{Csiszar1991} I. Csisz{\'a}r explained why information divergence is the only divergence function on the cone of positive measures that lead to tractable inference rules. One should observe that his inference rules are closely related to sufficiency and the Bregman identity, and the present paper may be view as a generalization of these results of I. Csisz{\'a}r.

\subsection{Statistical mechanics}

Statistical mechanics can be stated based on classical mechanics
or quantum mechanics. For our purpose this makes no difference
because our theorems are valid for both classical systems and
quantum systems. 

As we have seen before 
\begin{equation}
Ex=kT_{0}\cdot D\left(s\left\Vert
s_{0}\right.\right).\label{eq:Kelvin}
\end{equation}
Our general results for Bregman divergences imply that the
Bregman divergence based on this exergy satisfies
\[
D_{Ex}\left(s_{1},s_{2}\right)=kT_{0}\cdot
D\left(s_{1}\left\Vert s_{2}\right.\right).
\]
Therefore 
\[
D_{Ex}\left(\Phi\left(s_{1}\right),\Phi\left(s_{2}\right)\right)=D_{Ex}\left(s_{1},s_{2}\right)
\]
for any map that is sufficient for $\left\{
s_{1},s_{2}\right\} .$
The equality holds for any regret function that is reversible
and conserves the state that is in equilibrium with the environment.
Since a different temperature of the environment leads to a different
state that is in equilibrium the equality holds for any reversible
map that leave some equilibrium state invariant. We see that
$D_{Ex}\left(s_{1},s_{2}\right)$ is uniquely determined as long as there exists a sufficiently
large set of maps that are reversible.

In this exposition we have made some short-cuts. First of all we
did not derive equation \ref{eq:Kelvin}. In particular the notion of temperature
was used without discussion. Secondly we identified the internal
energy with the mean value of the Hamiltonian and identified the
thermodynamic entropy with $k$ times the Shannon entropy. Finally, in the
argument above we need to verify in all details that the set of
reversible maps is sufficiently large to determine the regret
function. For classical thermodynamics the most comprehensive exposition
was done by Lieb and Yngvason \cite{Lieb1998,Lieb2010}. In their
exposition randomness was not taken into account. With the present framework it is
also possible to handle randomness so that one can make a bridge between thermodynamics and
statistical mechanics. A detailed exposition will be given in a future
paper.

According to Equation (\ref{eq:Kelvin}) any bit of information
can be converted into an amount of energy! One may ask how this is
related to the mixing paradox (a special case of Gibbs' paradox).
Consider a container divided by a wall with a blue and a yellow gas on
each side of the wall. The question is how much energy can be
extracted by mixing the blue and the yellow gas?
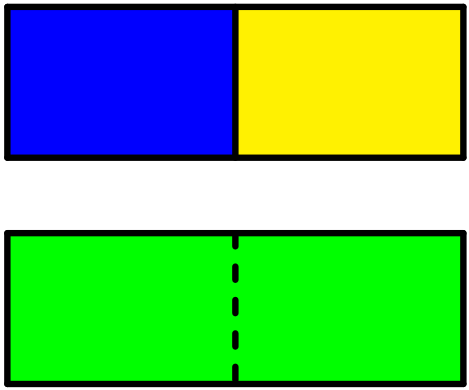
\begin{center}
\begin{tikzpicture}[line width=2.4pt, line cap=round,x=1.0cm,y=1.0cm]
\fill[fill=yellow,fill opacity=1.0] (8.,5.) -- (5.,5.) -- (5.,3.) -- (8.,3.) -- cycle;
\fill[fill=blue,fill opacity=1.0] (2.,5.) -- (5.,5.) -- (5.,3.) -- (2.,3.) -- cycle;
\fill[fill=green,fill opacity=1.0] (2.,2.) -- (8.,2.) -- (8.,0.) -- (2.,0.) -- cycle;
\draw  (2.,5.)-- (8.,5.);
\draw (8.,5.)-- (8.,3.);
\draw (8.,3.)-- (2.,3.);
\draw (2.,3.)-- (2.,5.);
\draw (5.,5.)-- (5.,3.);

\draw (2.,2.)-- (8.,2.);
\draw (8.,2.)-- (8.,0.);
\draw (8.,0.)-- (2.,0.);
\draw (2.,0.)-- (2.,2.);
\draw [dash pattern=on 5pt off 7.6pt] (5.,2.)-- (5.,0.);
\end{tikzpicture}
\par\end{center}

We loose one bit of information about each molecule by mixing the blue and the green gas, but if the color is the \emph{only difference} no energy
can be extracted. This seems to be in conflict with Equation
(\ref{eq:Kelvin}),
but in this case different states cannot be converted into each
other by reversible processes. For instance one cannot convert the
blue gas into the yellow gas. To get around this problem one can
restrict the set of preparations and one can restrict the set of
measurements.
For instance one may simply ignore measurements of the color of
the gas. What should be taken into account and what should be
ignored, can only be answered by an experienced physicist. Formally this
solves the mixing paradox, but from a practical point of view nothing
has been solved. If for instance the molecules in one of the gases
are much larger than the molecules in the other gas then a
semi-permeable membrane can be used to create an osmotic pressure that can be
used to extract some energy. It is still an open question which
differences in properties of the two gases that can be used to extract
energy.

\subsection{Monotone regret for portfolios}
We know that in general a local regret function on a state space with at least three orthogonal states is proportional to information divergence. In portfolio theory we get the stronger result that monotonicity implies that we are in the situation of gambling introduced by Kelly \cite{Kelly1956}.

\begin{thm}
\label{Theorem:proper}Assume that none of the assets
are dominated
by a portfolio of other assets. If the regret function $D_G (P,Q)$ given by (\ref{eq:Bregman})
is monotone then the regret function equals information divergence and
the measures $P$ and $Q$ are supported by $k$ distinct price
relative
vectors of the form $\left(o_{1},0,0,\dots0\right)$,
$\left(0,o_{2},0,\dots0\right),$
until $\left(0,0,\dots o_{k}\right).$ 
\end{thm}
\begin{proof}
If there are more than three price relative vectors then a monotone regret function is always proportional to information divergence  which is a strict regret function. Therefore we may assume that there are only two price relative vectors. Assume that the regret function is not strict. Then the function $G$ defined by (\ref{eq:G})  is not strictly convex. Assume that $D_G (P,Q)=0$ so that $G$ is affine between $P$ and $Q$. Let $\Phi$ be a contraction around one of the end points of intersection between the state space and the line through $P$ and $Q$. Then monotonicity implies that $D_G (\Phi (P),\Phi (Q))=0$ so that $G$ is affine on the line between $\Phi (P)$ and $\Phi (Q)$. This holds for contractions around any point. Therefore $G$ is affine on the whole state space which implies that there is a single portfolio that dominates all assets. Such a portfolio must be supported on a single asset. Therefore monotonicity implies that if two assets are not dominated then the regret function is strict and according to Theorem \ref{strict} we have already proved that a strict regret function in portfolio theory is proportional to information divergence.
\end{proof}

\begin{Example}
If the regret function divergence is monotone and one of the assets is the
safe asset then there exists a portfolio $\vec{b}$ such that
$b_{i}\cdot o_{i}\geq1$
for all $i.$ Equivalently $b_{i}\geq o_{i}^{-1}$ which is
possible if and only if $\sum o_{i}^{-1}\leq1.$ One say that the gamble
is \emph{fair} if $\sum o_{i}^{-1}=1$. If the gamble is
\emph{super-fair}, i.e. $\sum o_{i}^{-1}<1$, then the portfolio
$b_{i}=o_{i}^{-1}/\sum o_{i}^{-1}$
gives a price relative equal to $\left(\sum
o_{i}^{-1}\right)^{-1}>1$
independently of what happens, which is a \emph{Dutch book}. 
\end{Example}
\begin{cor}
In portfolio theory the regret function $D_G (P,Q)$ given by (\ref{eq:Bregman}) is monotone if and only if
it is strict.
\end{cor}
\begin{proof}
We use that in portfolio theory the regret function is monotone if and only it is proportional to information.
\end{proof}

\section{Concluding remarks}

In \cite{Pitrik2015} it was proved that if $f$ is a function such that the Bregman divergence based on $tr(f(\rho))$
is monotone on any (simple) C*-algebra then the Bregman divergence is jointly convex. As we have seen that monotonicity implies that the Bregman divergence must be proportional to inform divergence, which is jointly convex in both arguments. We also see that in general joint convexity is not a sufficient condition for monotonicity, but in the case where the state space has only two orthogonal states it is not known if joint convexity of a Bregman divergence is sufficient to conclude that the Bregman divergence is monotone.

One should note that the type of optimization presented in this paper is closely related to a game theoretic model developed by F. Tops\o e \cite{Topsoe2008,Topsoe2011}. In his game theoretic model he needed what he called the {\em perfect match principle}. Using the terminology presented in this paper the perfect match principle states that the regret function is a strict Bregman divergence. As we have seen the perfect match principle is only fulfilled in portfolio theory if all the assets are gambling assets. Therefore the theory of F. Tops\o e can only be used to describe gambling while our optimization model can describe general portfolio theory and our sufficient conditions can explain exactly when our general model equals gambling.

The original paper of Kullback and Leibler \cite{Kullback1951}
was called ``On Information and Sufficiency''. In the present paper
we have made the relation between information divergence and the
notion of sufficiency more explicit. The results presented in this
paper are closely related to the result that a divergence that
is both an $f$-divergence and
a Bregman divergence is proportional to information divergence
(see \cite{Harremoes2007a} or \cite{Amari2009} and references
therein). All $f$-divergences satisfy a sufficiency condition, which is the
reason why this class of divergences has played such a
prominent role in the study of
the relation between information theory and statistics. One
major question has been to find reasons for choosing between the
different $f$-divergences.
For instance $f$-divergences of power type (often called Tsallis
divergences or Cressie-Read divergences) are popular, but there
are surprisingly few papers that can point at a single value of the power $\alpha$
that is optimal for a certain problem except if this
value is 1. In this paper we have started with Bregman divergences because each optimization
problem comes with a specific Bregman divergence. Often it is
possible to specify a Bregman divergence for an optimization problem and only in some
of the cases this Bregman divergence is proportional to
information divergence.

The idea of sufficiency has different
relevance in different applications, but in all cases
information divergence prove to be the quantity that convert the general
notion of sufficiency into a number. In information theory information
divergence appear as a consequence of Kraft's inequality. For
code length functions of
integer length we get functions that are piecewise linear. Only
if we are interested in extend-able sequences we get a regret
function that satisfies a
data processing inequality. In this sense information theory is
a theory of extend-able sequences. For scoring functions in
statistics the notion of
locality is important. These applications do not refer to
sequences. Similarly the notion of sufficiency that plays a major
role in statistics, does not
refer to sequences. Both sufficiency and locality imply that
regret is proportional to information divergence, but these
reasons are different from the
reasons why information divergence is used in information
theory. Our description of statistical mechanics does not go
into technical details, but
the main point is that the many symmetries in terms of
reversible maps form a set of maps so
large that our result on invariance
of regret under sufficient maps applies. In this sense
statistical mechanics and statistics both apply information
divergence for reasons
related to sufficiency. For portfolio theory the story is
different. In most cases one has to apply the general theory of
Bregman divergences because we
deal with an optimization problem. The general Bregman
divergences only reduce to information divergence when the
assets are gambling assets.

Often one talk about applications of information theory in
statistics, statistical mechanics and portfolio theory. In this
paper we have argued that
information theory is mainly a theory of sequences, while some
problems in statistics and statistical mechanics are also
relevant without reference to
sequences. It would be more correct
to say that convex optimization has various application such as 
information theory, statistics,
statistical mechanics, and portfolio theory and that certain
conditions related to sufficiency lead to the same type of
quantities in all these applications.

\section*{Acknowledgment}

The author want to thank Prasad Santhanam for inviting me to the
Electrical Engineering Department, University of Hawai\textquoteleft i at
M\={a}noa, where many of the ideas presented in this paper were developed.
I also want to thank Alexander M{\"u}ller-Hermes, Frank Hansen, and Flemming Tops{\o}e
for stimulating discussions and correspondence.
Finally I want to thank the reviewers for their valuable
comments.


\begin{thebibliography}{-------}
\providecommand{\natexlab}[1]{#1}

\bibitem[Kullback and Leibler(1951)]{Kullback1951}
Kullback, S.; Leibler, R.
\newblock On Information and Sufficiency.
\newblock {\em Ann. Math. Statist.} {\bf 1951}, {\em 22},~79--86.

\bibitem[Jaynes(1957)]{Jaynes1957}
Jaynes, E.T.
\newblock Information Theory and Statistical Mechanics, {I} and {II}.
\newblock {\em Physical Reviews} {\bf 1957}, {\em 106 and 108},~620--630 and
  171--190.

\bibitem[Jaynes(1989)]{Jaynes1989}
Jaynes, E.T.
\newblock Clearing up mysteries -- The original goal. In {\em Maximum Entropy
  and {B}ayesian Methods}; Skilling, J., Ed.; Kluwer: Dordrecht,  1989.

\bibitem[Liese and Vajda(1987)]{Liese1987}
Liese, F.; Vajda, I.
\newblock {\em Convex Statistical Distances}; Teubner: Leipzig,  1987.

\bibitem[Barron \em{et~al.}(1998)Barron, Rissanen, and Yu]{Barron1998}
Barron, A.R.; Rissanen, J.; Yu, B.
\newblock The Minimum Description Length Principle in Coding and Modeling.
\newblock {\em IEEE Trans. Inform. Theory} {\bf 1998}, {\em 44},~2743--2760.
\newblock Commemorative issue.

\bibitem[Csisz{\'a}r and Shields(2004)]{Csiszar2004}
Csisz{\'a}r, I.; Shields, P.
\newblock {\em Information Theory and Statistics: A Tutorial}; Foundations and
  Trends in Communications and Information Theory, Now Publishers Inc.,  2004.

\bibitem[Gr{\"u}nwald and Dawid(2004)]{Grunwald2004a}
Gr{\"u}nwald, P.D.; Dawid, A.P.
\newblock Game Theory, Maximum Entropy, Minimum Discrepancy, and Robust
  {B}ayesian Decision Theory.
\newblock {\em {A}nnals of {M}athematical {S}tatistics} {\bf 2004}, {\em
  32},~1367--1433.

\bibitem[Gr{\"u}nwald(2007)]{Grunwald2007}
Gr{\"u}nwald, P.
\newblock {\em the Minimum Description Length principle}; MIT Press,  2007.

\bibitem[Holevo(1982)]{Holevo1982}
Holevo, A.S.
\newblock {\em Probabilistic and Statistical Aspects of Quantum Theory};
  Vol.~1, {\em North-Holland Series in Statistics and Probability},
  North-Holland: Amsterdam,  1982.

\bibitem[Krumm \em{et~al.}(2016)Krumm, Barnum, Barrett, and
  M{\"u}ller]{Krumm2016}
Krumm, M.; Barnum, H.; Barrett, J.; M{\"u}ller, M.
\newblock Thermodynamics and the structure of quantum theory.
\newblock arXiv:1608.04461.

\bibitem[Barnum \em{et~al.}(2014)Barnum, M{\"u}ller, and Ududec]{Barnum2014}
Barnum, H.; M{\"u}ller, M.P.; Ududec, C.
\newblock Higher-order interference and single-system postulates characterizing
  quantum theory.
\newblock {\em New Journal of Physics} {\bf 2014}, {\em 16},~123029.

\bibitem[Harremo{\"e}s(2016)]{Harremoes2016d}
Harremo{\"e}s, P.
\newblock Maximum Entropy and Sufficiency.
\newblock  Proceedings MaxEnt2016. American Institute of Physics (AIP),  2016,
  \href{http://xxx.lanl.gov/abs/arXiv:1607.02259}{{\normalfont
  [arXiv:1607.02259]}}.

\bibitem[Harremo{\"e}s(2017)]{Harremoes2017a}
Harremo{\"e}s, P.
\newblock Quantum information on Spectral Sets.
\newblock arXiv:1701.06688 Accepted for presentation at ISIT 2017.

\bibitem[Servage(1951)]{Servage1951}
Servage, L.J.
\newblock The Theory of Statistical Decision.
\newblock {\em Journal of the American Statistical Association} {\bf 1951},
  {\em 46},~55--67.

\bibitem[Kiwiel(1997{\natexlab{a}})]{Kiwiel1997}
Kiwiel, K.C.
\newblock Proximal Minimization Methods with Generalized Bregman Functions.
\newblock {\em SIAM Journal on Control and Optimization} {\bf 1997}, {\em
  35},~1142--1168,
  \href{http://xxx.lanl.gov/abs/http://dx.doi.org/10.1137/S0363012995281742}{{\normalfont
  [http://dx.doi.org/10.1137/S0363012995281742]}}.

\bibitem[Kiwiel(1997{\natexlab{b}})]{Kiwiel1997a}
Kiwiel, K.C.
\newblock Free-steering Relaxation Methods for Problems with Strictly Convex
  Costs and Linear Constraints.
\newblock {\em Math. Oper. Res.} {\bf 1997}, {\em 22},~326--349.

\bibitem[Rockafellar(1970)]{Rockafeller1970}
Rockafellar, R.T.
\newblock {\em Convex Analysis}; Princeton Univ. Press: New Jersey,  1970.

\bibitem[Hendrickson and Buehler(1971)]{Hendrickson1971}
Hendrickson, A.D.; Buehler, R.J.
\newblock Proper scores for probability forecasters.
\newblock {\em Ann. Math. Statist.} {\bf 1971}, {\em 42},~1916--1921.

\bibitem[Rao and Nayak(1985)]{Rao1985}
Rao, C.R.; Nayak, T.K.
\newblock Cross Entropy, Dissimilarity Measures, and Characterizations of
  Quadratic Entropy.
\newblock {\em IEEE Trans. Inform. Theory} {\bf 1985}, {\em 31},~589--593.

\bibitem[Banerjee \em{et~al.}(2005)Banerjee, Merugu, Dhillon, and
  Ghosh]{Banerjee2005}
Banerjee, A.; Merugu, S.; Dhillon, I.S.; Ghosh, J.
\newblock Clustering with {B}regman Divergences.
\newblock {\em Journal of Machine Learning Research} {\bf 2005}, {\em
  6},~1705--1749.

\bibitem[{McC}arthy(1956)]{McCarthy1956}
{McC}arthy, J.
\newblock Measures of the value of information.
\newblock {\em Proc. Nat. Acad. Sci.} {\bf 1956}, {\em 42},~654--655.

\bibitem[Gneiting and Raftery(2007)]{Gneiting2007}
Gneiting, T.; Raftery, A.E.
\newblock Strictly Proper Scoring Rules, Prediction, and Estimation.
\newblock {\em Journal of the American Statistical Association} {\bf 2007},
  {\em 102},~359--378,
  \href{http://xxx.lanl.gov/abs/http://dx.doi.org/10.1198/016214506000001437}{{\normalfont
  [http://dx.doi.org/10.1198/016214506000001437]}}.

\bibitem[Ovcharov(2015)]{Ovcharov2015}
Ovcharov, E.Y.
\newblock Proper Scoring Rules and Bregman Divergences.
\newblock Sept. 2015. arXiv:1502.01178.

\bibitem[Gundersen(2011)]{Gundersen2011a}
Gundersen, T.
\newblock An Introduction to the Concept of Exergy and Energy Quality.
\newblock Technical report, Department of Energy and Process Engineering,
  Norwegian University of Science and Technology, Trondheim, Norway,  2011.
\newblock
  http://www.ivt.ntnu.no/ept/fag/tep4120/innhold/Exergy

\bibitem[Harremo{\"e}s(1993)]{Harremoes1993}
Harremo{\"e}s, P.
\newblock {\em Time and Conditional Independence}; Vol. 255, {\em
  IMFUFA-tekst}, IMFUFA Roskilde University,  1993.
\newblock Original in Danish entitled Tid og Betinget Uafh{\ae}ngighed. English
  translation partially available.

\bibitem[Kelly(1956)]{Kelly1956}
Kelly, J.L.
\newblock A New Interpretation of Information Rate.
\newblock {\em Bell System Technical Journal} {\bf 1956}, {\em 35},~917--926.

\bibitem[Cover and Thomas(1991)]{Cover1991}
Cover, T.; Thomas, J.A.
\newblock {\em Elements of Information Theory}; Wiley,  1991.

\bibitem[Uhlmann(1970)]{Uhlmann1970}
Uhlmann, A.
\newblock On the {S}hannon Entropy and Related Functionals on Convex Sets.
\newblock {\em Reports on Mathematical Physics} {\bf 1970}, {\em 1},~147--159.

\bibitem[M{\"u}ller-Hermes and Reeb(2015)]{Mueller-Hermes2016}
M{\"u}ller-Hermes, A.; Reeb, D.
\newblock Monotonicity of the Quantum Relative Entropy Under Positive Maps.
\newblock {\em Annales Henri Poincare} {\bf 2015},
  \href{http://xxx.lanl.gov/abs/Sept. 2016. arXiv:1512.06117v2}{{\normalfont
  [Sept. 2016. arXiv:1512.06117v2]}}.

\bibitem[Christandl and M{\"u}ller-Hermes(2016)]{Christandl2016}
Christandl, M.; M{\"u}ller-Hermes, A.
\newblock Relative Entropy Bounds on Quantum, Private and Repeater Capacities.
\newblock April, 2016. arXiv:1604.03448.

\bibitem[Petz(2003)]{Petz2003}
Petz, D.
\newblock Monotonicity of Quantum Relative Entropy Revisited.
\newblock {\em Reviews in Mathematical Physics} {\bf 2003}, {\em 15},~79--91,
  \href{http://xxx.lanl.gov/abs/http://www.worldscientific.com/doi/pdf/10.1142/S0129055X03001576}{{\normalfont
  [http://www.worldscientific.com/doi/pdf/10.1142/S0129055X03001576]}}.

\bibitem[Petz(1988)]{Petz1988}
Petz, D.
\newblock Sufficiency of Channels over von {N}eumann algebras.
\newblock {\em Quart. J. Math. Oxford} {\bf 1988}, {\em 39},~97--108,.

\bibitem[Jen{\v c}ov{\'a} and Petz(2006)]{Jencova2006}
Jen{\v c}ov{\'a}, A.; Petz, D.
\newblock Sufficiency in quantum statistical inference.
\newblock {\em Communications in Mathematical Physics} {\bf 2006}, {\em
  263},~259--276.

\bibitem[Harremo{\"e}s and Tishby(2007)]{Harremoes2007a}
Harremo{\"e}s, P.; Tishby, N.
\newblock The Information Bottleneck Revisited or How to Choose a Good
  Distortion Measure.
\newblock  Proceedings ISIT 2007, Nice. IEEE Information Theory Society,  2007,
  pp. 566--571.

\bibitem[Jiao \em{et~al.}(2014)Jiao, amd Albert~No, Venkat, and
  Weissman]{Jiao2014}
Jiao, J.; amd Albert~No, T.C.; Venkat, K.; Weissman, T.
\newblock Information Measures: the Curious Case of the Binary Alphabet.
\newblock {\em Trans. Inform. Theory} {\bf 2014}, {\em 60},~7616--7626.

\bibitem[Jen{\v c}ov{\'a}(2017)]{Jencova2017}
Jen{\v c}ov{\'a}, A.
\newblock Preservation of a quantum {R}{\'e}nyi relative entropy implies
  existence of a recovery map.
\newblock {\em Journal of Physics A: Mathematical and Theoretical} {\bf 2017},
  {\em 50},~085303.

\bibitem[Tishby \em{et~al.}(1999)Tishby, Pereira, and Bialek]{Tishby1999}
Tishby, N.; Pereira, F.; Bialek, W.
\newblock The information bottleneck method.
\newblock  Proceedings of the 37-th Annual Allerton Conference on
  Communication, Controland Computing,  1999, pp. 368--377.

\bibitem[No and Weissman(2015)]{No2015}
No, A.; Weissman, T.
\newblock Universality of logarithmic loss in lossy compression.
\newblock  2015 IEEE International Symposium on Information Theory (ISIT),
  2015, pp. 2166--2170.

\bibitem[Dawid \em{et~al.}(2012)Dawid, Lauritzen, and Perry]{Dawid2012}
Dawid, A.P.; Lauritzen, S.; Perry, M.
\newblock Proper local scoring rules on discrete sample spaces.
\newblock {\em The Annals of Statistics} {\bf 2012}, {\em 40},~593--603.

\bibitem[Bernardo(1978)]{Bernardo1978}
Bernardo, J.M.
\newblock Expected Information as Expected Utility.
\newblock {\em The Annals of Statistics} {\bf 1978}, {\em 7},~686--690.
\newblock Institute of Mathematical Statistics.

\bibitem[Csisz{\'a}r(1991)]{Csiszar1991}
Csisz{\'a}r, I.
\newblock Why least squares and maximum entropy? An axiomatic approach to
  inference for linear inverse problems.
\newblock {\em Ann. Stat.} {\bf 1991}, {\em 19},~2032--2066.

\bibitem[Lieb and Yngvason(1998)]{Lieb1998}
Lieb, E.; Yngvason, J.
\newblock A Guide to Entropy and the Second Law of Thermodynamics.
\newblock {\em Notices of the AMS} {\bf 1998}, {\em 45},~571--581.

\bibitem[Lieb and Yngvason(2010)]{Lieb2010}
Lieb, E.; Yngvason, J., The Mathematics of the Second Law of Thermodynamics.
\newblock In {\em Visions in Mathematics}; Alon, N.; Bourgain, J.; Connes, A.;
  Gromov, M.; Milman, V., Eds.; Birkh{\"a}user Basel,  2010; pp. 334--358.

\bibitem[Pitrik and Virosztek(2015)]{Pitrik2015}
Pitrik, J.; Virosztek, D.
\newblock On the Joint Convexity of the Bregman Divergence of Matrices.
\newblock {\em Letters in Mathematical Physics} {\bf 2015}, {\em
  105},~675--692.

\bibitem[Tops{\o}e(2008)]{Topsoe2008}
Tops{\o}e, F.
\newblock Game theoretical optimization inspired by information theory.
\newblock {\em Journal of Global Optimization} {\bf 2008}, {\em 43},~553.

\bibitem[Tops{\o}e(2011)]{Topsoe2011}
Tops{\o}e, F.
\newblock Cognition and Inference in an Abstract Setting.
\newblock  Proceedings WITMSE 2011,  2011.

\bibitem[Amari(2009)]{Amari2009}
Amari, S.I.
\newblock $\alpha$-Divergence Is Unique, Belonging to Both $f$-Divergence and
  Bregman Divergence Classes.
\newblock {\em IEEE Transactions on Information Theory} {\bf 2009}, {\em
  55},~4925--4931.

\end{thebibliography}


\end{document}